\documentclass[epj,nopacs]{svjour}

\usepackage{graphicx}
\usepackage{amsmath,amssymb,amsfonts,mathtools}
\usepackage{bm}

\usepackage{booktabs}
\usepackage{tabularx}
\usepackage{array}
\usepackage{enumitem}

\usepackage{microtype}
\usepackage{xcolor}
\usepackage{hyperref}

\hypersetup{
    colorlinks=true,
    linkcolor=blue,
    citecolor=blue,
    urlcolor=blue
}

\newcommand{\rank}{\operatorname{rank}}
\newcommand{\Ker}{\operatorname{Ker}}
\newcommand{\im}{\operatorname{Im}}
\newcommand{\cl}{\operatorname{cl}}
\newcommand{\R}{\mathbb{R}}
\newcommand{\CP}{\mathbb{CP}}
\newcommand{\Apr}{A_{\rm primary}}
\newcommand{\Xmicro}{X_{\rm micro}}
\newcommand{\Xprimary}{X_{\rm primary}}
\newcommand{\Lprimary}{L_{\rm primary}}
\newcommand{\Dprimary}{\mathcal D_{\rm primary}}
\newcommand{\Piobs}{\Pi_{\rm obs}}
\newcommand{\Madm}{\mathcal M_{\rm adm}}
\newcommand{\Vnull}{V_{\rm phys\text{-}null}}
\newcommand{\dop}{d_{\rm op}}
\newcommand{\deff}{d_{\rm eff}}

\begin{document}

\title{Compositional Operational Rank, Smooth Realizability, and Dimensional Reduction}

\author{
Nikolay D. Tulupov\inst{1}%
\thanks{ORCID: \href{https://orcid.org/0009-0004-2017-2288}
{0009-0004-2017-2288}}
\and
Alexander M. Tishin\inst{1,2}%
\thanks{ORCID: \href{https://orcid.org/0000-0003-2252-7279}
{0000-0003-2252-7279}}
}

\titlerunning{Compositional Operational Rank and Dimensional Reduction}
\authorrunning{N.D. Tulupov and A.M. Tishin}

\institute{
Phystech School of Radio Engineering and Computer Technology,
Moscow Institute of Physics and Technology,
Institutsky Lane 9, Dolgoprudny, 141700, Moscow Region, Russia
\and
Physics Faculty,
Lomonosov Moscow State University,
Leninskie Gory 1, Moscow, 119991, Russia
}

\mail{tulupov.nd@phystech.edu}

\date{}

\abstract{
Higher-dimensional theories distinguish the dimensionality of a microscopic
description from the number of directions accessible to a physical observer.
Relating these notions requires not only dimensional reduction, but also a
consistent preservation of operational independence across different levels
of description.
We establish sufficient conditions under which a compositional operational
rank is preserved first as the dimension of a smooth integral leaf and then
as an observable Fisher-information rank. The key requirement is
closure-faithfulness: the smooth realization must preserve the dependency
relations of the operational closure, rather than merely its total rank.
Together with real representability, constant rank, involutivity, and
observer/scale compatibility, this yields a local integral leaf whose
dimension equals the operational rank.
A further primary-identification condition is required at the measurement
level. If an independently specified observer image coincides locally with
the primary integral distribution and the admissible protocols provide local 
resolution, finite labeled-mixture closure, and locally constant outcome 
support, then the same rank can be realized by a finite composite 
Fisher-information protocol. The construction holds for arbitrary finite 
rank. At $r=4$ it provides a conditional operational criterion for 
dimensional reduction in higher-dimensional theories: a viable reduction 
must preserve the relevant operational dependency structure, not merely 
produce four effective coordinates. Quantum-hardware negative controls 
further show that the same restricted Fisher rank can conceal different 
larger ranks under a broader pre-specified measurement protocol. 
The framework therefore constrains the operational outcome of dimensional 
reduction without imposing a universal four-dimensional bound on microscopic 
state spaces.
\keywords{
Operational distinguishability,
compositional operational rank,
closure-faithfulness,
Fisher information,
dimensional reduction,
higher-dimensional theories
}
}

\maketitle

\section{Introduction}

Higher-dimensional theories distinguish two notions that need not coincide:
the dimensionality of a microscopic or formal description and the number of
independent directions accessible to a physical observer. In perturbative
string theory, internal quantum consistency selects $D=26$ for the bosonic
string and $D=10$ for superstrings \cite{polchinski,gsw}, while eleven
dimensions play a distinguished role in M-theory and in strong-coupling
relations among the ten-dimensional superstring theories
\cite{witten1995}. Standard approaches reconcile such higher-dimensional
descriptions with four-dimensional macroscopic physics through
compactification, localization, dynamical decoupling, or related mechanisms
\cite{coudarchet2024,apers2025}.

A reduction of geometric dimension, however, does not by itself determine
which independent physical distinctions survive in the reduced theory.
Microscopic directions may become inaccessible, redundant, or
indistinguishable for different physical reasons, while an incomplete
measurement may fail to resolve a direction that remains present.
Consequently, neither a four-dimensional effective description nor a measured
rank-four Fisher matrix is by itself sufficient to identify a
four-dimensional operational structure. The question addressed here is what
additional structural conditions are required for such an identification.

Three notions of rank enter this problem. An operational rank counts
elementary procedures that remain independent under finite composition. A
smooth rank counts independent tangent directions, for example the dimension
of an integral leaf of a distribution. An experiment probes the rank of a
Fisher information matrix, which counts parameter directions that produce
locally distinguishable changes in outcome statistics \cite{amari2016}.
These quantities are defined at different levels of description and need not
agree. In particular, a smooth realization can preserve total rank while
changing dependency relations among subsets of transitions, and a restricted
measurement can underestimate the locally accessible rank.

We establish sufficient conditions under which these three ranks can be
related without identifying them by assumption. The construction is valid for
an arbitrary finite operational rank $r$. Its central structural requirement
is closure-faithfulness: the smooth realization must preserve the
compositional dependency relations of the operational closure, rather than
merely its total rank. Together with real representability, regularity,
involutivity, and compatibility of the physical regime, this yields a local
integral leaf whose dimension equals the operational rank. We show explicitly
that total-rank agreement is insufficient, that closure-faithfulness and
involutivity are independent conditions, and that the criterion can be
preserved by some reductions while failing in non-Abelian compositional
realizations.

A further identification is required at the measurement level. We introduce
an observer architecture in which the microscopic description is mapped to a
primary observer sector before an admissible measurement protocol is applied.
The observer map is specified independently of the measured Fisher spectrum,
and its tangent image must coincide locally with the primary integral
distribution. Under this identification, local resolution by the admissible
protocol family, finite labeled-mixture closure, and locally constant outcome
support allow its resolving power to be represented by a finite Fisher matrix
of the same rank. Thus the operational, smooth, and Fisher ranks coincide only
under an explicit sequence of structural and measurement assumptions.

The physically relevant specialization considered here is $r=4$, but the
value four is not derived by the bridge. It is supplied by the pre-geometric
construction of Ref.~\cite{before}. There, requirements C1--C4 select the
representation dimension
\begin{equation}
 C1\text{--}C4 \Longrightarrow n=4,
 \label{eq:c1c4}
\end{equation}
whereas its identification with the closure-based operational rank requires
the additional Hypothesis H,
\begin{equation}
 H:\qquad \rank A([s])=n.
 \label{eq:H}
\end{equation}
This identification is Conjecture 3.1a of Ref.~\cite{before}, not a
consequence of C1--C4. With the additional requirement PR0 that the primary
observer sector itself satisfy the assumptions of that construction, the
rank-four specialization is therefore conditional:
\begin{equation}
 \mathrm{PR0}+H
 \Longrightarrow
 \rank \Apr([s])=4.
 \label{eq:chain2}
\end{equation}
The general bridge developed below does not depend on the validity of this
rank-four construction; it preserves whatever finite primary rank is supplied
to it.

The companion distinguishability-geometry framework \cite{dg} likewise does
not impose a four-dimensional microscopic state space. It introduces a smooth
distinguishability geometry on state spaces of arbitrary dimension, with an
effective distinguishability dimension that may depend on observational
scale. The representation dimension of Eq.~\eqref{eq:c1c4}, the primary
operational rank, the dimension of a smooth microscopic state space, and a
microscopic quantum Fisher rank must therefore remain distinct unless the
corresponding physical identifications are independently established.

This distinction changes the interpretation of dimensional reduction. A
$D$-dimensional microscopic or target-space theory is not inconsistent because
$D>4$. Rather, if it is to realize the conditional rank-four primary sector
above, it must provide a physically specified map to an observer-accessible
sector and preserve the operational dependency structure required by the
bridge. Compactification is one possible realization of such a map, but not
the unique one. Section~\ref{sec:operational-reduction} develops this criterion
for higher-dimensional theories. We show that the correct reduced dimension
does not by itself guarantee closure-faithfulness, examine a benign
Kaluza--Klein circle reduction, study surviving Wilson-line directions, and
exhibit non-Abelian compositional realizations in which the
operational-to-smooth bridge can fail.

A separate quantum-hardware negative control probes the measurement end of
the construction. Two closely related six-parameter circuit families have the
same restricted Fisher rank four, while the same pre-specified enlarged
measurement protocol reveals rank five in one family and rank six in the
other. The distinction is reproduced on three IBM Quantum processors and
remains robust under changes of shot count and calibration epoch. These
experiments do not measure primary operational rank and do not provide
evidence for fundamental four-dimensionality. Their role is to demonstrate
experimentally that measurement completeness cannot be inferred from a
restricted Fisher spectrum.

Figure~\ref{fig:conceptual-architecture} summarizes the architecture of the
argument. The compositional operational rank, its smooth realization, and the
observable Fisher rank belong to successive levels of description and coincide
only under the stated bridge and completeness assumptions. In parallel, a
higher-dimensional microscopic theory reaches an effective low-dimensional
sector through a separate physical reduction mechanism. Identifying that
sector with the primary operational sector is an additional physical step,
not a consequence of either construction alone.

\begin{figure*}[!t]
    \centering
    \includegraphics[width=\textwidth]{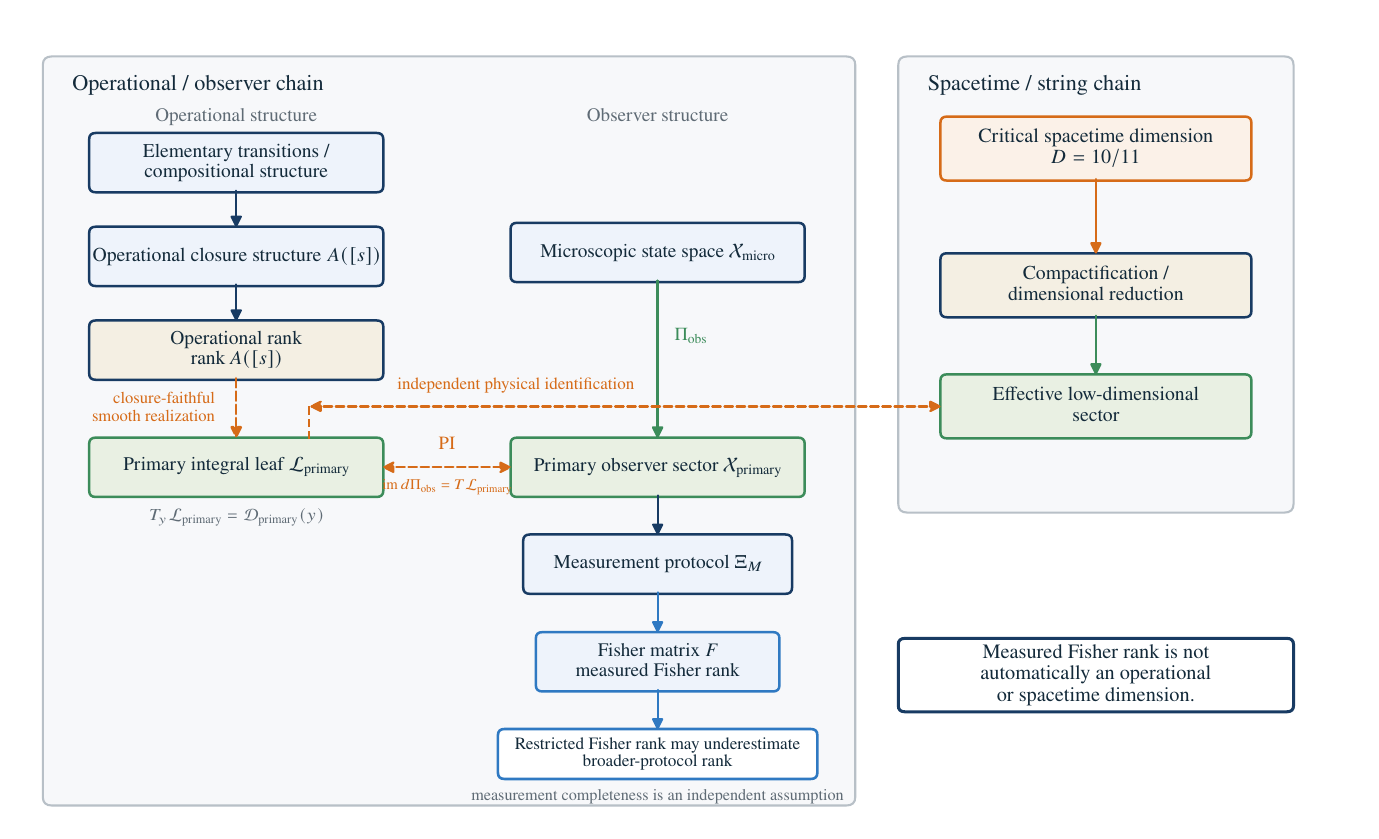}
    \caption{
    Conceptual architecture of the operational and spacetime reductions considered
    in this work. On the operational side, the compositional structure defines the
    operational closure structure $A([s])$ and its operational rank. A
    rank-preserving smooth realization requires real representability,
    preservation of the compositional dependency relations, constant rank,
    involutivity, and compatibility of the compared physical regime. An
    independently specified observer/pointer map and a locally resolving
    measurement protocol then determine the observable Fisher rank; incomplete
    measurements can only reduce the resolved rank. On the spacetime side, a
    critical microscopic dimension such as $D=10$ or $11$ is related to an
    effective low-dimensional sector through compactification or another
    dimensional-reduction mechanism. The identification of the resulting effective
    spacetime sector with the primary smooth operational sector is an independent
    physical assumption. Accordingly, a measured Fisher rank is not by itself an
    operational rank or a spacetime dimension
    }
    \label{fig:conceptual-architecture}
\end{figure*}

\section{Operational accessibility}
\label{sec:accessibility}

Let a smooth physical description be locally parametrized by coordinates
$q^A$, $A=1,\ldots,N$. The mere existence of $N$ formal coordinates does not
establish that all $N$ directions are independently observable. Operational
accessibility depends on what physical procedures are admissible and on which
distinctions survive the observer map.

This motivates the separation
\begin{equation}
 \Xmicro \xrightarrow{\Piobs} \Xprimary \xrightarrow{\Xi_M} \Delta_M,
 \label{eq:observer-architecture}
\end{equation}
where $\Xmicro$ is a microscopic or formal state space, $\Xprimary$ is the
observer/pointer sector relevant to the primary distinguishability structure,
and $\Delta_M$ is the probability simplex of the outcomes of protocol $M$.
The map $\Xi_M:\Xprimary\to\Delta_M$ assigns to each primary state its outcome
distribution under that protocol.
The resulting observable statistical map is
\begin{equation}
 \Phi_M:\Xmicro\longrightarrow\Delta_M,
 \qquad
 \Phi_M=\Xi_M\circ\Piobs.
 \label{eq:factorization}
\end{equation}

The central requirement is non-circularity. Neither the primary sector
$\Xprimary$ nor the observer map $\Piobs$ may be selected after examining the
Fisher spectrum in order to force a desired rank.
They must instead be fixed independently of that spectrum by the
physical observer structure, stable records, a specified coarse-graining, an
observable algebra, or an explicit open-system dynamics
\cite{breuer2016,schlosshauer2019}.

\section{Compositional operational rank}
\label{sec:operational-rank}

\subsection{The accessibility structure and its rank}

The definitions collected here are those of Ref.~\cite{before}, restated so
that
the remainder of the paper can be read independently of it. They define only a
discrete structure and its rank; no dimensional result or consistency
statement
from Ref.~\cite{before} is used in this subsection. 
Let $S$ denote the set of physically realizable states, with operational
indistinguishability relation $\sim$. For a local indistinguishability class
$[s]$, admissible procedures generate a local transition structure
\begin{equation}
 A([s])=(O([s]),T[s]).
 \label{eq:A}
\end{equation}
The operational dependence relation is compositional: a transition is
dependent on a set of other transitions if it can be reproduced by a finite
composition of the corresponding procedures.

For the present purpose it is useful to write this in closure form. If
$E([s])$ denotes a finite set of elementary transitions, let $\cl(X)=\langle
X\rangle$ denote the compositional closure of a subset $X\subseteq E([s])$,
that is, the set of transitions reproducible from $X$. The operational rank is
the maximal cardinality of a jointly non-redundant set with respect to this
closure.

\paragraph{Standing finite-rank assumption.}
For the general results of Sections~\ref{sec:smooth}--\ref{sec:observer} we
require only a finite set $E([s])$ of elementary transitions equipped with the
compositional closure $\cl$, of finite operational rank $r$.

Finiteness of the rank is then automatic, since a jointly non-redundant subset
of a finite set is finite. Nothing further about the origin of the pair
$(E([s]),\cl)$ is used, and nothing in this subsection fixes a value for $r$.
Ref.~\cite{before} obtains such a setting from two postulates, recalled in
Section~\ref{sec:rankfour} because they are needed for the instantiation
developed there and for comparison with that construction; they are not
assumptions of the general bridge once the pair is supplied.

\subsection{The instantiation at rank four}
\label{sec:rankfour}

Ref.~\cite{before} obtains the finite procedural setting used in its rank-four
construction from two additional postulates, recalled here for completeness.
They belong to that construction and are not assumptions of the general bridge
of Sections~\ref{sec:smooth}--\ref{sec:observer}, which needs only the
standing assumption above.

\paragraph{Axiom F (Finite resolution).}
For any indistinguishability class $[s]$ and any admissible procedure
$P\in\mathcal P$, the set of operationally distinguishable outcomes of $P$
applied to representatives of $[s]$ is finite.

\paragraph{Axiom G (Local finite generation).}
For any indistinguishability class $[s]$, the monoid $\mathcal P([s])$ of
admissible procedures is finitely generated: there exists a finite set
$G\subset\mathcal P([s])$ such that every element of $\mathcal P([s])$ is a
finite composition of elements of $G$.

Here $\mathcal P$ denotes the global set of physically admissible procedures,
of which a single procedure is written $P$, while $\mathcal P([s])$ is the
local monoid at $[s]$. In Ref.~\cite{before} these axioms supply the finite
procedural setting from which the rank-four construction is developed. Finite
generation by $G$ supplies a generating set rather than a basis. Whether the
minimal generating dimension agrees with the closure-based operational rank is
precisely the content of Hypothesis H.

Separately, a minimum-cardinality generating set $G$ of the local monoid
$\mathcal P([s])$ defines the associated free real vector space
\begin{equation}
 \mathbb V([s])=\R G,
 \qquad n=\dim\mathbb V([s])=|G|.
 \label{eq:V}
\end{equation}
Theorem 3.1 of Ref.~\cite{before} selects the minimal representation dimension
\begin{equation}
 n=4
 \label{eq:n4}
\end{equation}
under C1--C4. The identification of $n$ with $\rank A([s])$ is not
definitional and is not a consequence of C1--C4. It is Hypothesis H. In the
terminology of Ref.~\cite{before} the conditional layer is stated as
Conjecture 3.1a, yielding the emergent effective dimension $d_{\rm eff}=4$.

One inequality is available without H. Lemma A.5.2a of Ref.~\cite{before}
gives $\rank A([s])\ge n$ unconditionally, so C1--C4 already place a lower
bound of four on the operational rank, and Hypothesis H is what converts that
bound into an equality. The distinction propagates: without H, but under the
remaining smooth-realization assumptions, Section~\ref{sec:smooth-realization}
yields $\dim\Lprimary\ge4$, and under the corresponding observer and
measurement assumptions Corollary~\ref{cor:t1c} yields $\rank F_{\rm
primary}^{(M_{\rm comp})}\ge4$. Thus the falsifiable content of the
unconditional layer is a floor rather than a value.

The rank-four specialization depends on its premises in a way that matters
later. It requires $\Apr([s])$ itself to satisfy Axioms F and G and
C1--C4, and does not assert that a structure obtained by reduction from a
microscopic structure with those properties inherits them.
Section~\ref{sec:operational-reduction} takes this up as a separate
requirement on reduction mechanisms, stated there as PR0.

\begin{theorem}[T1a: conditional primary rank-four specialization]
\label{thm:t1a}
Let the primary local distinguishability structure $\Apr([s])$ satisfy Axioms
F and G, requirements C1--C4, and Hypothesis H. Then
\begin{equation}
 \rank \Apr([s])=4.
 \label{eq:t1a}
\end{equation}
\end{theorem}

\begin{remark}
T1a is not a statement about the dimension of a Hilbert space or the quantum
Fisher rank of an arbitrary family of states.
\end{remark}

\begin{remark}
The value four is used in what follows only as an input to the bridge. 
Lemma~\ref{lem:fisher}, Lemma~\ref{lem:composite}, Theorem~\ref{thm:t1b} and
Theorem~\ref{thm:observer} hold for an arbitrary primary rank $r$;
replacing $4$ by $r$ throughout Sections~\ref{sec:smooth}--\ref{sec:observer}
changes no argument.
The bridge conditions can therefore be assessed independently of how the
primary rank is fixed.
\end{remark}

Because Ref.~\cite{before} is presently a preprint, the logical dependence
should be explicit. Theorem~\ref{thm:t1a}, and with it the value four, is
conditional on that construction. By contrast, the bridge results developed
from Definition~\ref{def:closurefaithful} onward require only the standing
finite-rank assumption, except where the rank-four specialization is invoked
explicitly. They therefore remain valid for arbitrary finite rank $r$
independently of the rank-four instantiation.

\subsection{Physical status of Hypothesis H}
\label{sec:Hstatus}

Because every rank-four equality after Eq.~\eqref{eq:chain2} is conditional on
H, and because Ref.~\cite{before} establishes that H is neither provable nor
refutable from Axioms F and G and requirements C1--C5, it is worth setting out
here what H asserts physically, independently of the algebra that produced it.

The two numbers it identifies are constructed differently. The representation
dimension $n$ is the cardinality of a minimum-cardinality generating set $G$
of the local monoid $\mathcal P([s])$: the number of elementary procedures
from which every admissible local procedure can be composed. The operational
rank $\rank A([s])$ is the maximal cardinality of a jointly independent family
of elementary transitions under the closure $\cl$. For vector spaces the two
counts agree, and the habit of treating them as one number comes from that
setting; in more general algebraic structures they need not. For finitely
generated monoids they need not agree: irredundant generating sets of
different cardinalities can coexist, a phenomenon documented in the semigroup
literature \cite{howie1995} and used in Ref.~\cite{before} to prove that H is
not derivable. What does hold unconditionally is the inequality $\rank
A([s])\ge n$.

Physically, H excludes an accidental mismatch between procedural generation
and operational independence. What is established mathematically is only the
inequality above; one mechanism by which it can be strict is that a composite
of several elementary distinctions happens to coincide with another admissible
procedure for reasons unrelated to the physics being distinguished, so that
the procedures can be generated more economically than the distinctions they
draw can be counted. The corresponding failure mode is therefore strict
inequality, $\rank A([s])>n$: the closure structure contains more jointly
independent transitions than are represented by the cardinality of a minimal
generating set of procedures. A counterexample to H would be a physical system
whose local monoid of admissible procedures exhibits exactly this redundancy,
invisible at the level of individual transitions and appearing only in
composites.

A sufficient condition is known. If the lattice $L([s])$ of closed subsets of
$E([s])$ is geometric, that is, finite, atomistic, and semimodular
\cite{oxley}, then H holds; the record model of Ref.~\cite{before}, whose
additional assumptions make $L([s])$ Boolean, is the simplest realization.
Neither condition is claimed to be necessary, and neither is established for
physically realizable structures in general.

One consequence should be recorded explicitly, because it bears on
Section~\ref{sec:experiment}. Since $\rank A([s])\ge n$ holds without H, a
negative experiment can challenge the H-based rank-four identification without
challenging the unconditional layer $C1$--$C4\Rightarrow n=4$. For that
inference, however, it is not enough merely to observe a Fisher rank above
four somewhere in a larger observer image. At least five measured tangent
directions must have been specified independently as belonging to the same
primary operational realization, and the relevant operational-to-tangent
identification must be closure-faithful on that set. Under those conditions a
robust measured Fisher rank above four is incompatible with the equality
supplied by H.

Global informational completeness is not required for this one-sided
inference, because an incomplete protocol can only lower the measured rank.
What such an observation would not do is establish a value of the primary
rank. A positive identification requires the full bridge conditions of
Section~\ref{sec:smooth-realization}, PI, local resolution, closure under
finite labeled mixtures, and locally constant outcome support. The
falsification protocol proposed below is therefore a test of Hypothesis H
together with the operational-to-smooth identification used in the experiment,
not a test of the pre-geometric representation result itself.

\section{Smooth operational rank and Fisher information}
\label{sec:smooth}

\subsection{Protocol-independent null directions}

This section is formulated for an arbitrary smooth parameter manifold $Q$ and
is applied below with several instantiations. For $Q=\Xmicro$ we write $V_{\rm
null}^{\rm micro}$, whereas for $Q=\Xprimary$ or $Q=\Lprimary$ we write
$V_{\rm null}^{\rm primary}$. When a statement applies to any of these
instantiations, we use the unadorned notation $\Vnull(q)$ defined below.
Nothing in Sections~\ref{sec:smooth}--\ref{sec:observer} depends on the chosen
instantiation unless explicitly stated.

Let $Q$ be a smooth $N$-dimensional parameter manifold and let $\Madm$ be the
set of admissible protocols within one fixed observer sector, scale, and local
regime. Each protocol defines a smooth statistical map
\begin{equation}
 \Phi_M:Q\to\Delta_M.
 \label{eq:Phi}
\end{equation}
We define the fully operationally null tangent space by
\begin{equation}
 \Vnull(q):=\bigcap_{M\in\Madm}\Ker d\Phi_{M,q},
 \label{eq:vnull}
\end{equation}
and the protocol-independent smooth operational dimension by
\begin{equation}
 \dop(q):=\dim\!\left(T_qQ/\Vnull(q)\right).
 \label{eq:dop}
\end{equation}
This definition separates a physical null direction from a direction that is
invisible only to one incomplete measurement.

\subsection{Fisher rank and local identifiability}

Let an admissible protocol $M$ have outcome probabilities $p_i(q)$ with
locally constant support $I(q)=\{i:p_i(q)>0\}$. The classical Fisher
information matrix, defined on that support, is
\begin{equation}
 F^{(M)}_{AB}(q)
 =\sum_{i\in I(q)}\frac{1}{p_i(q)}
 \frac{\partial p_i}{\partial q^A}
 \frac{\partial p_i}{\partial q^B}.
 \label{eq:fisher}
\end{equation}
Writing $J=d\Phi_{M,q}$ and
$W=\operatorname{diag}(p_i(q)^{-1})$ gives
$F=J^{\mathsf T}WJ$.
The quantum counterpart of this construction and the associated statistical
distance are standard \cite{braunstein,liu2020,pezze2018}.

\begin{lemma}[Fisher rank and local identifiability]
\label{lem:fisher}
If the statistical weight is non-degenerate on $\im J$, then
\begin{equation}
\begin{aligned}
 \Ker F^{(M)}(q)
 &=\Ker d\Phi_{M,q},\\
 \rank F^{(M)}(q)
 &=\rank d\Phi_{M,q}.
\end{aligned}
\label{eq:fisher-rank}
\end{equation}

More generally, for positive semidefinite $W$,
\begin{equation}
\begin{aligned}
 &\rank(J^{\mathsf T}WJ)=\rank J\\
 &\qquad\Longleftrightarrow\quad
 \Ker W\cap\im J=\{0\}.
\end{aligned}
\end{equation}
\end{lemma}

With locally constant support, the non-degeneracy condition is automatic
rather than an additional physical assumption: restricting the sum in
Eq.~\eqref{eq:fisher} to $I(q)$ makes $W$ positive definite, so $\Ker W=\{0\}$
and the second condition holds trivially. Points at which the support changes
fall outside the scope of Lemma~\ref{lem:fisher}, and with it outside
Lemma~\ref{lem:composite} and Corollary~\ref{cor:t1c}; they are treated as an
experimental complication in Section~\ref{sec:experiment}.

Equation~\eqref{eq:fisher-rank} is the differential form of the standard
local-identifiability criterion. It contains no four-dimensional bound by
itself.

\subsection{Composite protocols}

A useful simplification occurs if $\Madm$ is closed under finite labeled
mixtures. A labeled mixture chooses protocol $M_i$ with positive probability
$\lambda_i$ and records the label $i$ together with the outcome.

For parameter-independent weights $\lambda_i>0$, the Fisher matrix of the
labeled mixture is
\begin{equation}
 F^{(M_{\rm comp})}(q)
 =\sum_i \lambda_i F^{(M_i)}(q),
 \label{eq:mixture-fisher}
\end{equation}
and, since each term is positive semidefinite,
\begin{equation}
 \Ker F^{(M_{\rm comp})}(q)
 =\bigcap_i \Ker F^{(M_i)}(q).
 \label{eq:mixture-kernel}
\end{equation}

\begin{lemma}[Composite protocol]
\label{lem:composite}
Suppose $\dim T_qQ=N<\infty$, $\Madm$ is closed under finite labeled
mixtures, and the non-degeneracy condition of Lemma~\ref{lem:fisher} holds at
$q$ for every admissible protocol used in the construction. In particular,
this condition is automatic on a region of locally constant outcome support.
Then there exists a finite composite protocol $M_{\rm comp}$ such that
\begin{equation}
 \Ker F^{(M_{\rm comp})}(q)=\Vnull(q),
 \label{eq:compkernel}
\end{equation}
and therefore
\begin{equation}
 \rank F^{(M_{\rm comp})}(q)
 =\dop(q)
 =\sup_{M\in\Madm}\rank F^{(M)}(q).
 \label{eq:comprank}
\end{equation}
The last equality follows because $\Vnull(q)\subseteq\Ker F^{(M)}(q)$ for
every admissible protocol, while closure under finite labeled mixtures makes
the rank-$\dop(q)$ composite itself an element of $\Madm$.

If $\dop(q)>0$, at most
\begin{equation}
 \dop(q)=N-\dim\Vnull(q)
\end{equation}
component protocols are required. 
If $\dop(q)=0$ every admissible protocol is locally null and the statement is
trivial.
\end{lemma}

To see the bound on the number of component protocols, start from $T_qQ$.
Whenever the current kernel intersection is strictly larger than $\Vnull(q)$,
the definition of $\Vnull(q)$ guarantees an admissible protocol whose kernel
removes at least one remaining direction. Thus no more than $\dop(q)$ such
steps are required.

The scope of this result should be stated precisely, because it is easy to
overread. Closure under labeled mixtures creates no new distinguishable
direction: if every protocol in $\Madm$ is blind to a tangent vector $v$, so
is every labeled mixture of them. Lemma~\ref{lem:composite} therefore
compresses the resolving power already present in a family into a single
finite composite protocol with the same null space; it does not establish that
the family resolves the sector, and in particular it does not yield
$\Vnull(q)=0$. Informational completeness remains a separate assumption about
the family, imposed where it is needed in Section~\ref{sec:observer}.

\section{Why an arbitrary Lie-orbit bridge is insufficient}
\label{sec:lie}

A natural attempt is to represent elementary distinguishability transitions by
generators of a Lie algebra acting on a smooth state space. This is useful but
not rank preserving in general.

Let a connected Lie group $G_{\rm access}$ with Lie algebra
$\mathfrak g_{\rm access}$ act smoothly on a manifold $X$.
The infinitesimal action at $x$ is
\begin{equation}
 \alpha_x:\mathfrak g_{\rm access}\to T_xX.
 \label{eq:alpha}
\end{equation}
The orbit dimension is
\begin{equation}
 \dim(G_{\rm access}\cdot x)
 =\dim\mathfrak g_{\rm access}-\dim\mathfrak g_{\rm stab}(x),
 \label{eq:orbitdim}
\end{equation}
where $\mathfrak g_{\rm stab}(x)$ is the stabilizer algebra at $x$.
For a smooth map $\Psi:Q\to G_{\rm access}\cdot x$,
\begin{equation}
 \rank d\Psi_q\le \dim(G_{\rm access}\cdot\Psi(q)).
 \label{eq:maprank}
\end{equation}

Several distinct invariants may therefore be associated with the same
construction: the minimal number of Lie-algebra generators
$\mu(\mathfrak g)$, the algebra dimension $\dim\mathfrak g$, the orbit rank
$r_{\rm orbit}:=\dim(G_{\rm access}\!\cdot x)$, the differential rank
$r_{\rm map}:=\rank d\Psi_q$, the Fisher rank, and the compositional
operational rank:
\begin{equation}
\begin{aligned}
 &\mu(\mathfrak g),\quad
 \dim\mathfrak g,\quad
 r_{\rm orbit},\quad
 r_{\rm map},\\
 &\rank F^{(M)},\quad
 \rank A([s]).
\end{aligned}
\label{eq:invariants}
\end{equation}
There is no universal equality between them.

The obstruction is bracket growth. If elementary candidates are
$\xi_1,\ldots,\xi_r$, their Lie closure contains commutators such as
$[\xi_a,\xi_b]$ and higher brackets, which may create new tangent directions
even though the corresponding composite transition is operationally dependent
on the original generators. The three-dimensional Heisenberg algebra
\cite{hall2015} provides the simplest counterexample. It is generated as a Lie
algebra by two elements $X$ and $Y$, with $[X,Y]=Z$, so that $\mu(\mathfrak
h_3)=2$ while $\dim\mathfrak h_3=3$. For the corresponding Heisenberg group
acting on itself by left translations, the orbit is three-dimensional. The
same example does more than illustrate bracket growth, and it is taken up
again in Section~\ref{sec:nonabelian} as an explicit realization in which the
distribution is involutive while closure-faithfulness fails.

Thus the Lie algebra rank condition is a useful diagnostic of dynamical
accessibility \cite{dalessandro}, while a restricted Fisher or quantum Fisher
information matrix diagnoses local statistical distinguishability within the
chosen measurement or state family \cite{braunstein,liu2020,pezze2018}.
Neither diagnostic establishes that compositional closure is faithfully
represented by linear or Lie closure.

\section{Rank-preserving smooth realization}
\label{sec:smooth-realization}

\subsection{Closure-faithfulness and real representability}

Let
\begin{equation}
 \vartheta_x:E([s])\to T_x\Xprimary
 \label{eq:vartheta}
\end{equation}
map elementary operational transitions to tangent vectors of the primary
manifold. The realization is taken into the full tangent space and not into a
quotient of it. Passing to a quotient by an observer-null subspace before
closure-faithfulness has been imposed would allow two operationally
independent transitions to acquire dependent images, and the identification of
distribution rank with operational rank in Theorem~\ref{thm:t1b} would fail.
The observer quotient therefore enters later, in Section~\ref{sec:observer},
and no measurement structure is used in the present section.

\begin{definition}[Closure-faithful realization]
\label{def:closurefaithful}
The realization is closure-faithful if for every $X\subseteq E([s])$ and $e\in
E([s])$,
\begin{equation}
 e\in\cl(X)
 \quad\Longleftrightarrow\quad
 \vartheta_x(e)\in\operatorname{span}\vartheta_x(X).
 \label{eq:closurefaithful}
\end{equation}
\end{definition}

Closure-faithfulness is stronger than equality of total ranks. It requires the
same dependency relations for every subset of elementary transitions. Stated
that way the condition ranges over $2^m$ subsets; in fact far less is needed.

\begin{proposition}[Closure-faithfulness is checkable on small subsets]
\label{prop:smallsubsets}
Let the accessibility structure on $E([s])$ be matroidal of rank $r$, and let
$\vartheta_x$ be a realization whose images span an $r$-dimensional subspace.
If, for every subset $Y\subseteq E([s])$ with $|Y|\le r+1$, operational
dependence of $Y$ agrees with linear dependence of
$\vartheta_x(Y)$, then $\vartheta_x$ is closure-faithful on all of $E([s])$.
\end{proposition}

\begin{proof}
Write $M$ for the accessibility matroid and $M'$ for the linear matroid of the
images. Every circuit of $M$ has at most $r+1$ elements, since removing one
element of a circuit leaves an independent set, of size at most $r$.
By the hypothesis on Eq.~\eqref{eq:closurefaithful}, a subset $Y$ with
$|Y|\le r+1$ is dependent in $M$ if and only if its image is linearly
dependent in $M'$. Hence $M$ and $M'$ agree on dependency for all subsets
of cardinality at most $r+1$; in particular $M'$ has no independent set of
size $r+1$, so $\rank M'\le r$ and the circuits of $M'$ are likewise of size
at most $r+1$. The two matroids therefore have the same circuits. A set is
dependent exactly when it contains a circuit, so $M=M'$ and
Eq.~\eqref{eq:closurefaithful} holds throughout.
\end{proof}

At the rank of interest here this bound is five: closure-faithfulness of a
four-dimensional realization is decided by its behaviour on subsets of at most
five elementary transitions. The bound is not improvable in general, but
failures are often visible much earlier; in
Proposition~\ref{prop:rankdoesnotdiscriminate} below a single pair already
separates a faithful realization from an unfaithful one of the same rank.

\begin{corollary}[Uniform case, and a discriminant]
\label{cor:uniform}
If the accessibility structure is the uniform matroid $U_{r,m}$ and the
images of $\vartheta_x$ span an $r$-dimensional subspace, then
closure-faithfulness of $\vartheta_x$ is equivalent to linear independence
of the images of every $r$-element subset of $E([s])$. Equivalently, after
fixing the distinguishability metric $G$ of Ref.~\cite{dg}, the corresponding
Gram determinant is nonzero for every $r$-element subset.
For each $r$-element subset $S\subseteq E([s])$, define
\begin{equation}
c_S(x)=
\left\{
\begin{aligned}
&\displaystyle
\frac{
\sqrt{
\det\!\left[
G(\vartheta_x(e_a),\vartheta_x(e_b))
\right]_{a,b\in S}
}
}{
\displaystyle
\prod_{a\in S}\|\vartheta_x(e_a)\|_G
},
\\[-0.3ex]
&\hspace{1em}
\text{if }\displaystyle
\prod_{a\in S}\|\vartheta_x(e_a)\|_G>0,
\\[1ex]
&0,
\qquad
\text{if }\displaystyle
\prod_{a\in S}\|\vartheta_x(e_a)\|_G=0.
\end{aligned}
\right.
\label{eq:normalized-gram-volume}
\end{equation}
The closure-faithfulness discriminant is then
\begin{equation}
 \mathcal C(x)
 =
 \min_{\substack{S\subseteq E([s])\\ |S|=r}}
 c_S(x).
 \label{eq:closure-discriminant}
\end{equation}
By Hadamard's inequality,
\[
 0\le c_S(x)\le1,
 \qquad
 0\le\mathcal C(x)\le1.
\]
Under the hypotheses of the corollary, $\mathcal C(x)=0$ if and only if
at least one $r$-element subset has linearly dependent images, including
the case in which one of those images vanishes. Equivalently,
closure-faithfulness fails.

Proposition~\ref{prop:free} below shows that the structures arising in the
simplest commuting physical realizations are free matroids $U_{m,m}$ rather
than the non-trivial uniform structures $U_{r,m}$ with $r<m$ considered here.
\end{corollary}

The reference to $G$ is essential. The vanishing of the Gram determinant, and
hence the locus on which closure-faithfulness fails, is independent of
coordinates. The numerical value of $\mathcal C(x)$, however, depends on the
chosen distinguishability metric $G$ and therefore becomes meaningful only
after that metric has been specified.

There is also a representability issue. Even if the operational accessibility
lattice defines a matroid, it does not follow that this matroid is
representable over $\R$ \cite{oxley}. We therefore require:

\paragraph{B2-R (real representability).}
The closure structure of elementary transitions in the chosen primary sector
admits a real linear representation for which operational closure coincides
with linear-span closure.

B2-R and RB3 below play different logical roles. B2-R is an
existence/admissibility condition on the operational closure structure: it
excludes closure structures for which no real linear representation is
possible at all. RB3 is a condition on the particular smooth realization used
in the bridge, requiring that this chosen realization actually preserve the
operational closure. Thus RB3 exhibits a real representation in the region
where it holds, while B2-R records separately that real representability is a
genuine restriction on the class of operational structures under
consideration.

Logically, RB3 for a given real smooth realization already witnesses B2-R
locally. We retain B2-R separately because it is an admissibility condition on
the operational structure prior to the construction of any particular smooth
realization, whereas RB3 tests the realization actually used in the bridge.

Two sufficient classes are already available from the constructions discussed
in the companion works. Proposition~\ref{prop:free} below shows that a monoid
generated by commuting one-parameter subgroups under the stated independence
hypothesis has a free accessibility matroid, and a free matroid is
representable over $\R$, so B2-R holds for the whole of that class without
further assumption. The record model discussed in Ref.~\cite{before} is a
second such case: the additional record-model assumptions of
Ref.~\cite{before} make the accessibility lattice $L([s])$ the Boolean lattice
$B_n$ on $n$ atoms, again the lattice of flats of a free matroid.

\subsection{Primary distribution}

Let $X_1,\ldots,X_m$ be smooth vector fields realizing the elementary
transitions $e_1,\ldots,e_m$ of $E([s])$ on a smooth primary manifold, so that
the realization above is given pointwise by $\vartheta_x(e_i)=X_i(x)$.
Smoothness of the family $\vartheta_x$ is then inherited from the fields and
need not be assumed separately. Define
\begin{equation}
 \Dprimary(x)=\operatorname{span}\{X_1(x),\ldots,X_m(x)\}.
 \label{eq:Dprimary}
\end{equation}
We impose the following sufficient conditions:

\begin{description}
\item[RB1: constant rank.] $\Dprimary$ has constant rank $r$ in a regular
neighborhood.
\item[RB2: involutivity.] $[X,Y]\in\Gamma(\Dprimary)$ for all local sections
$X,Y$ of $\Dprimary$.
\item[RB3: closure-faithfulness.] Equation~\eqref{eq:closurefaithful} holds in
the regular region for the chosen smooth realization.
\item[RB4: observer/scale compatibility.]
The discrete structure and its smooth realization refer to the same primary
observer sector, physical scale, and local regime as the admissible protocols
to be introduced in the measurement stage.
\end{description}

These six conditions---a physical identification of the elementary transitions
with the chosen fields, B2-R, RB1, RB2, RB3 and RB4---are sufficient for the
bridge. They are not claimed to be necessary, and no attempt is made here to
show that the set is minimal. RB2 and RB3 in particular are independent of one
another despite standing adjacent: the first concerns integrability of the
distribution as a whole, the second whether compositional dependence among the
elementary transitions matches linear dependence of their images.
Section~\ref{sec:nonabelian} exhibits a realization satisfying RB1 and RB2 in
which RB3 fails.

\begin{theorem}[T1b: primary smooth bridge] 
\label{thm:t1b}
Assume that the primary operational structure $\Apr([s])$ has finite
rank $r$, and that its elementary transitions are physically identified with
the chosen smooth fields. Assume also B2-R and RB1--RB4. Then through every 
regular point there exists an integral leaf $\Lprimary$ satisfying
\begin{equation}
 \dim\Lprimary=\rank\Apr([s])=r.
 \label{eq:t1b}
\end{equation}
\end{theorem}

\begin{proof}
By RB1 and RB2, the Frobenius theorem \cite{lee} integrates $\Dprimary$ to
local leaves whose tangent spaces equal the distribution. RB3 holds at every
point of the regular region and, because $\vartheta_x$ takes values in the
full tangent space, it identifies the linear independence relations among
$X_1(x),\ldots,X_m(x)$ with the operational closure relations among
$e_1,\ldots,e_m$; hence the rank of $\Dprimary(x)$ equals $\rank\Apr([s])$.
RB4 guarantees that the compared structures belong to the same physical
sector.
\end{proof}

For the T1a specialization, the theorem gives a four-dimensional primary
leaf, $\dim\Lprimary=4$. This is the case used in the dimensional-reduction
application of Section~\ref{sec:higher}. Nothing in the proof distinguishes 
the value four from any other finite rank.

\section{Observer/pointer quotient and Fisher realization}
\label{sec:observer}

The smooth bridge does not imply that the full microscopic state space is
four-dimensional. The physical interpretation adopted here is a two-stage
architecture
\begin{equation}
 \Xmicro\xrightarrow{\Piobs}\Xprimary\xrightarrow{\Xi_M}\Delta_M.
 \label{eq:observer2}
\end{equation}
The observer map may be realized, depending on the theory, as a smooth
coarse-graining,
a completely positive trace-preserving (CPTP) channel followed by
coordinates on its image, a conditional expectation onto an
accessible observable algebra, a map to stable pointer records, or an
effective channel induced
by a specified decoherence dynamics \cite{breuer2016,schlosshauer2019}.
No one of these mechanisms is built into the definition of operational rank.

Two comments on the status of this construction are needed before the theorem.
First, $\Piobs$ is a physical input and not a free mathematical projection.
Nothing in the framework licenses choosing it so as to produce a desired rank:
it must be fixed by the observer structure, the record algebra, a specified
coarse-graining, or an explicit open-system dynamics of the theory under
consideration, and fixed before any Fisher spectrum is examined. The theorem
below is conditional on such an input in exactly the sense that
Theorem~\ref{thm:t1b} is conditional on a physically justified identification
of the elementary transitions.

Second, Theorem~\ref{thm:observer} isolates the consequences of the observer
and local-resolution assumptions. Once all admissible protocols factor
through $\Piobs$ and locally resolve its image, the protocol-independent null
space coincides with the kernel of the observer differential. The substantive
rank-preservation step is therefore Theorem~\ref{thm:t1b}, while
Lemma~\ref{lem:composite} provides a finite compression of resolving power
already present in the protocol family. Local completeness remains an
independent assumption.

We say that the primary protocol family locally resolves the image of
$\Piobs$ at $y=\Piobs(x)$ if
\begin{equation}
 \bigcap_{M\in\Madm}
 \Ker\!\left(
 d_y\Xi_M\big|_{\im d_x\Piobs}
 \right)
 =
 \{0\}.
 \label{eq:localresolution}
\end{equation}

\begin{theorem}[Observer-quotient realization]
\label{thm:observer}
Let $\Piobs:\Xmicro\to\Xprimary$ have constant rank $r$ in a regular
neighborhood. Suppose all admissible protocols factor through $\Piobs$ as in
Eq.~\eqref{eq:factorization} and that the primary family $\{\Xi_M\}$ locally
resolves the image of $\Piobs$. Then, in the instantiation $Q=\Xmicro$ of
Section~\ref{sec:smooth},
\begin{equation}
 \bigcap_{M\in\Madm}\Ker d_x\Phi_M
 =\Ker d_x\Piobs,
 \label{eq:observerkernel}
\end{equation}
and
\begin{equation}
 \dop(x)=\rank d_x\Piobs=r.
 \label{eq:observerrank}
\end{equation}
\end{theorem}

\begin{proof}
Let $y=\Piobs(x)$. Since
\[
d_x\Phi_M
=
d_y\Xi_M\circ d_x\Piobs,
\]
factorization implies
\[
\Ker d_x\Piobs
\subseteq
\bigcap_{M\in\Madm}\Ker d_x\Phi_M.
\]
Conversely, local resolution of the image of $\Piobs$ implies that any
tangent vector whose image under $d_x\Piobs$ is nonzero is detected by at
least one admissible primary protocol. Hence
\[
\bigcap_{M\in\Madm}\Ker d_x\Phi_M
=
\Ker d_x\Piobs.
\]
Taking codimensions in $T_x\Xmicro$ gives
\[
\dop(x)=\rank d_x\Piobs=r.
\]
\end{proof}

\begin{remark}[Three uses of the word primary]
\label{rem:threeprimary}
The adjective is applied to three a priori distinct objects: the discrete
operational structure $\Apr([s])$ of Section~\ref{sec:operational-rank}, its
smooth integral realization $\Dprimary$ and $\Lprimary$ of
Section~\ref{sec:smooth-realization}, and the observer image
$\Xprimary=\Piobs(\Xmicro)$ of the present section. Theorem~\ref{thm:t1b}
connects the first two and Theorem~\ref{thm:observer} characterizes the third,
but nothing so far identifies the tangent sector resolved by the observer map
with the leaf produced by the Frobenius theorem. That identification is not
definitional. Nothing in the statements above prevents $\dim\Lprimary=4$ while
$\dim\Xprimary=6$ and $\rank d_x\Piobs=6$: each theorem remains true, a
closure-faithful four-dimensional leaf exists inside $\Xprimary$, and the
observer map nevertheless resolves two further directions, so the conclusion
that the observer-accessible primary rank is four does not follow. RB4 does
not close the gap, since it requires only that the compared structures belong
to one local regime. The missing condition is stated next.
\end{remark}

\paragraph{PI (primary identification).}
In a regular neighborhood, and with $y=\Piobs(x)$,
\begin{equation}
 \im d_x\Piobs=\Dprimary(y)=T_y\Lprimary.
 \label{eq:PI}
\end{equation}

Under the constant-rank hypothesis of Theorem~\ref{thm:observer}, this means
that $\Piobs$ is locally a submersion onto the primary integral leaf and that
the locally realized observer image is an open subset of that leaf. PI is a
physical requirement on the pairing of the observer map with the primary
distribution, not a consequence of either. It is the last of the
identifications on which the chain depends, and like $\Piobs$ itself it must
be established independently of the measured rank.

Combining Eqs.~\eqref{eq:t1b}, \eqref{eq:comprank}, and
\eqref{eq:observerrank} gives the Fisher realization once PI is imposed.
Without PI, the ranks appearing in Eqs.~\eqref{eq:t1b} and
\eqref{eq:observerrank} belong to different objects, and no equality between
them follows.

The route is worth stating explicitly. Theorem~\ref{thm:t1b} gives
$\dim\Lprimary=r$, and PI identifies this tangent sector with $\im d_x\Piobs$,
so that $\rank d_x\Piobs=r$. Local resolution of the admissible primary family
makes the protocol-independent null space on the leaf trivial.
Lemma~\ref{lem:composite} then supplies a finite composite protocol with the
same null space, and Lemma~\ref{lem:fisher}, on a region of locally constant
support, converts this statement into Fisher rank $r$. Local completeness
remains an independent assumption; closure under labeled mixtures guarantees
only that the resolving power already present in the family can be compressed
into a finite composite protocol.

\begin{corollary}[T1c: restricted Fisher rank]
\label{cor:t1c}
Assume Theorem~\ref{thm:t1b}, Theorem~\ref{thm:observer}, and PI.
Suppose the admissible primary protocols locally resolve the primary leaf,
are closed under finite labeled mixtures, and have locally constant outcome
support in the region considered. Then
\begin{equation}
 \rank F_{\rm primary}^{(M_{\rm comp})}=r.
 \label{eq:t1c}
\end{equation}

Every primary protocol satisfies
\begin{equation}
 \rank F_{\rm primary}^{(M)}\le r.
 \label{eq:t1c-incomplete}
\end{equation}
Strict inequality may occur for incomplete protocols.
At $r=4$ this is the statement used below. The microscopic quantum Fisher rank
may nevertheless exceed the primary rank, whatever its value.
\end{corollary}

For an $N$-level pure quantum system the
ambient projective space is $\CP^{N-1}$ with real dimension $2N-2$
\cite{bengtsson2017}. Thus two-qubit and three-qubit pure-state manifolds have
real dimensions $6$ and $14$, respectively. Generic quantum families can
therefore have Fisher or quantum Fisher rank greater than four without
contradicting T1a--T1c.

\section{Operational constraints on dimensional reduction in higher-dimensional theories}
\label{sec:higher}

\subsection{Critical dimensions are not operational ranks}

Higher-dimensional theories naturally distinguish the dimensionality of a
microscopic or formal description from the number of independent directions
accessible within a specified physical observer sector. In perturbative string
theory the target-space dimension is fixed by internal quantum consistency:
cancellation of the world-sheet conformal anomaly selects $D=26$ for the
bosonic string and $D=10$ for superstrings \cite{polchinski,gsw}, while eleven
dimensions play a distinguished role in the non-perturbative framework whose
low-energy limit is eleven-dimensional supergravity \cite{witten1995}.
These dimensions characterize the microscopic formalism. By contrast, the
operational rank of Section~\ref{sec:operational-rank} characterizes the
independence structure of procedures available within an observer-accessible
sector. Neither invariant determines the other.

The problem considered here is therefore not to derive the critical dimension
of a higher-dimensional theory or to replace compactification dynamics, but to
identify the structural conditions under which a microscopic theory of
dimension $N$ can realize a lower-rank observer-accessible sector. Matching
the reduced number of tangent directions is not sufficient: the reduction
must also preserve the operational dependency structure to which that rank is
assigned. The remainder of this section formulates this requirement and tests
it in explicit classes of dimensional reduction.

\subsection{Observer-relative null directions}

Let the microscopic description be locally $N$-dimensional and let the
physically specified observer map have differential rank $r$ in the regular
region of interest. The rank-four application considered below corresponds to
$r=4$.

\begin{corollary}[T2: inaccessible microscopic directions]
\label{cor:t2}

If $\dim T_x\Xmicro=N$ and $\rank d_x\Piobs=r$, then
\begin{equation}
 \dim\Ker d_x\Piobs=N-r.
 \label{eq:nullNminusR}
\end{equation}

For every admissible protocol that factors through $\Piobs$, and on a region
of locally constant outcome support,
\begin{equation}
\begin{aligned}
 \Ker d_x\Phi_M
 &\supseteq \Ker d_x\Piobs,\\
 \rank F^{(M)}
 &\le \rank d_x\Phi_M\le r.
\end{aligned}
\label{eq:factorrank}
\end{equation}

The kernel in Eq.~\eqref{eq:nullNminusR} is observer-relative. Its elements
are microscopic tangent directions that are removed by the specified observer
map, but they should not be identified automatically with compact
coordinates, gauge directions, absent physical degrees of freedom, or any
particular microscopic mechanism.

\end{corollary}

Corollary~\ref{cor:t2} fixes only the number of microscopic directions removed
by the specified observer map. Different physical mechanisms may produce the
same value of $\rank d_x\Piobs$ while inducing different operational structures
in the reduced sector.

\subsection{Operational reduction requirement}
\label{sec:operational-reduction}

The central implication is not that operational distinguishability selects a
particular compactification mechanism, but that any mechanism producing the
primary observer sector must satisfy an additional structural requirement.

\begin{proposition}[T3: operational reduction requirement]
\label{prop:t3}

Let a theory possess a local microscopic or formal description with $N$
independent directions, and let its physically realized primary sector have
operational rank $r<N$ and satisfy the conditions of
Theorem~\ref{thm:t1b}, Theorem~\ref{thm:observer}, and PI. Then the theory
must provide a physically defined observer map
\begin{equation}
 X_{\rm micro}^{(N)}
 \xrightarrow{\Piobs}
 X_{\rm primary},
 \qquad
 \rank d_x\Piobs=r,
 \label{eq:reductionmap}
\end{equation}
whose differential therefore satisfies
\begin{equation}
 \dim\Ker d_x\Piobs=N-r.
 \label{eq:reductionkernel}
\end{equation}

Thus exactly $N-r$ microscopic tangent directions are absent from primary
operational accessibility at the regular point considered. In the
instantiation of Section~\ref{sec:rankfour} the relevant case is $r=4$, which
is the specialization used throughout the remainder of this section.

\end{proposition}

The counting requirement in Proposition~\ref{prop:t3} is necessary but not
sufficient. A further requirement enters before a reduced sector can be
identified with the rank-four operational structure of
Section~\ref{sec:rankfour}. The premise of T1a is that $\Apr([s])$ itself
satisfies Axioms F and G and requirements C1--C4. These properties are not
automatically inherited from a microscopic theory merely because a map of
state spaces has been specified.

\paragraph{PR0 (primary-structure admissibility).}
The reduction must induce, or independently specify, a primary operational
structure $\Apr([s])$ whose admissible procedures and closure relation satisfy
Axioms F and G and requirements C1--C4 in the observer sector and at the scale
under consideration.

A map between state spaces does not by itself specify how admissible
procedures and their compositional closure descend. Such a descent would
require additional structure at the level of procedures. The individual
requirements can behave differently under reduction. Closure under
composition may survive when the reduction is compatible with composition,
whereas non-degeneracy can fail when distinct microscopic distinctions merge.
Minimality is not generally inherited, because the reduced sector is tested
within a different class of admissible procedures. Inter-observer
consistency, which constrains $\operatorname{Aut}(C_D)$, likewise does not
follow from a state-space quotient alone. No general descent theorem is
assumed here.

A successful dimensional reduction must therefore supply more than the
numerical equality $\rank d_x\Piobs=r$. In addition to PR0, the rank-preserving
interpretation requires Hypothesis H when equality is claimed, a real
closure-faithful smooth realization satisfying RB1--RB4, and PI identifying
the observer image with the resulting primary distribution. Experimental
identification of the rank further requires an admissible protocol family that
locally resolves that sector.

These conditions distinguish a reduction of operational structure from a mere
reduction in coordinate number. Possible realizations include compactification,
brane localization \cite{maartens2010,alencar2025}, gauge redundancy, 
dynamical decoupling or freezing of massive modes, projection, superselection
\cite{bartlett2007}, and decoherence-induced coarse-graining
\cite{schlosshauer2019}. They act at different stages of the operational chain
and need not preserve the same structure. Operational distinguishability
therefore does not select a unique reduction mechanism; it constrains the
operational structure that an admissible reduction must leave behind.

\subsection{Implications for string-theoretic compactifications}

For a $D$-dimensional target-space or field-configuration description, the
rank-four statement is necessarily conditional. If the theory admits a
physically specified observer-accessible primary sector that satisfies PR0,
Hypothesis H, the smooth rank-preserving bridge, and PI, then that sector has
primary operational rank four. The remaining microscopic directions need not
be absent from the theory; rather, they must fail to define independent
primary distinguishability directions in the observer sector and at the scale
for which the rank-four statement is asserted.

A separate physical map is therefore required,
\begin{equation}
 \substack{\text{target-space /}\\
           \text{field configurations}}
 \xrightarrow{\Piobs}
 \substack{\text{primary observer}\\
           \text{sector}},
 \label{eq:stringPi}
\end{equation}
before any microscopic coordinate or field direction can be identified with
an element of $\Ker d_x\Piobs$. In particular, $\Piobs$ should not be read as
a generic synonym for compactification. Depending on the theory and the
physical regime, it may encode the effective loss of access to compact
directions, massive modes, gauge-redundant variables, localized sectors, or
other degrees of freedom. What matters for the present criterion is that the
map and the corresponding primary sector are specified independently of the
rank subsequently measured there.

The criterion supplements rather than replaces standard compactification
dynamics. A higher-dimensional construction may determine which low-energy
fields and directions survive, while the present question is whether the
resulting observer-accessible sector preserves the compositional dependency
structure required for its reduced rank to have an operational interpretation.
No inconsistency of a $D>4$ microscopic or target space is implied; the
rank-four statement concerns only a primary sector satisfying the assumptions
of Section~\ref{sec:rankfour}. The companion framework \cite{dg} likewise
allows state spaces of arbitrary dimension and scale-dependent accessible
distinguishability.

The following subsections compare standard reduction mechanisms, show that the
correct reduced rank does not guarantee closure-faithfulness, and examine
circle compactification, surviving Wilson-line directions, and non-Abelian
compositional generation.

\subsection{Where reduction mechanisms enter the operational bridge}

The mechanisms listed after Proposition~\ref{prop:t3} are not operationally
equivalent. Although each can reduce the number of accessible directions, they
act at different stages of the chain
\eqref{eq:observer-architecture} and therefore stress different conditions of
the rank-preserving bridge. Table~\ref{tab:mechanisms} summarizes this
correspondence.

\begin{table*}[htbp]
\centering
\caption{
    Representative reduction mechanisms, the level at which they remove
    directions from primary operational accessibility, and the bridge condition
    that becomes relevant in each case
}
\label{tab:mechanisms}
\renewcommand{\arraystretch}{1.25}
\begin{tabularx}{\textwidth}{>{\raggedright\arraybackslash}p{0.22\textwidth}
>{\raggedright\arraybackslash}p{0.34\textwidth} X}
\toprule
\textbf{Mechanism} & \textbf{Level at which directions are removed} &
\textbf{Principal bridge issue}\\
\midrule
Compactification &
Accessible coordinate and mode sector below the compactification scale &
RB4: the reduction is scale dependent and the discarded directions become
accessible again above the Kaluza--Klein threshold\\

Brane localization &
Support of the admissible procedures &
RB1: the rank of the accessible distribution may change near a localization
boundary\\

Gauge redundancy &
Passage from redundant field or configuration variables to the physical
state space &
B2-R and RB3: real representability and closure-faithfulness must be verified
on the physical quotient rather than on the redundant description\\

Superselection &
Accessible observable algebra &
RB3: the operational closure must be evaluated within the specified
superselection sector and observable algebra\\

Decoupling or freezing of massive modes &
Dynamics at accessible energies &
RB4: the reduction is energy dependent rather than structural\\

Decoherence-induced coarse-graining &
Stable pointer or record sector &
$\Piobs$, PI, and RB4: the observer image, physical scale, and environmental
regime must be specified independently\\
\bottomrule
\end{tabularx}
\end{table*}

Table~\ref{tab:mechanisms} highlights two points. First, loss of accessible
directions may be exact, as for gauge redundancy or exact superselection, or
regime dependent, as for compactification, decoupling, localization, and
decoherence. Second, the reduced tangent dimension alone does not determine
whether the surviving realization is closure-faithful. Two mechanisms can
therefore produce the same effective rank while preserving different
operational dependency structures. The next subsection gives an explicit
example.

\subsection{Correct reduced rank does not imply closure-faithfulness}

Obtaining the correct reduced rank does not guarantee the operational
criterion. A quotient can create new linear dependencies among tangent
representatives, even though it cannot remove dependencies that were already
present.

\begin{lemma}[Quotients create dependencies]
\label{lem:quotients}

Let $\vartheta_x:E([s])\to T_x\Xprimary$ be a realization and let
$P:T_x\Xprimary\to W$ be linear. Then for every $X\subseteq E([s])$ and every
$e\in E([s])$,
\begin{equation}
\begin{aligned}
 \vartheta_x(e)
 &\in \operatorname{span}\vartheta_x(X)\\
 &\Longrightarrow
 P\vartheta_x(e)
 \in \operatorname{span}P\vartheta_x(X).
\end{aligned}
\label{eq:quotientdep}
\end{equation}
The converse implication fails in general.

\end{lemma}

\begin{proof}
A linear map carries spans into spans, which proves the stated implication.
Failure of the converse is exhibited explicitly in
Proposition~\ref{prop:rankdoesnotdiscriminate}.
\end{proof}

Thus an observer quotient can introduce additional linear dependencies among
the images of elementary transitions. If the microscopic realization is
closure-faithful, the forward implication of
Eq.~\eqref{eq:closurefaithful} survives the quotient automatically, whereas
the reverse implication may fail. This reverse implication is precisely what
prevents a tangent dependency created by the reduction from being mistaken
for a pre-existing operational dependency. Total rank alone cannot detect
such a failure.

The following construction shows that this issue persists even when the
reduced dimension has exactly the desired value.

\begin{proposition}[A correct rank does not discriminate]
\label{prop:rankdoesnotdiscriminate}

There exist two linear quotients of the same closure-faithful rank-five
microscopic realization that both have linear rank four relative to the same
rank-four primary closure structure, while only one is closure-faithful.

\end{proposition}

\begin{proof}
Let $E([s])=\{e_1,\ldots,e_5\}$ carry two distinct closure structures. The
microscopic closure $\cl_{\rm micro}$ is the free one, under which every
subset is independent, so the microscopic rank is five. The primary closure
$\cl_{\rm primary}$ is that of the uniform matroid $U_{4,5}$, in which every
four-element subset is independent and every five-element subset is dependent,
so the primary operational rank is four.

The primary closure is not inferred from the quotient. By the non-circularity
requirement of Section~\ref{sec:accessibility}, it is fixed independently by
the primary observer sector. The question is whether a given quotient provides
a faithful realization of that independently specified structure.

Take a microscopic realization in $\mathbb{R}^5$ with basis vectors
$f_1,\ldots,f_5$. Define
\[
\rho(e_i)=f_i,\qquad i=1,\ldots,5.
\]
This realization is closure-faithful for $\cl_{\rm micro}$ and has rank five.
Consider the two linear maps $P_1,P_2:\R^5\to\R^4$ that act as the identity on
$f_1,\ldots,f_4$ and satisfy
\begin{equation}
 P_1f_5=f_1,
 \qquad
 P_2f_5=f_1+f_2+f_3+f_4.
 \label{eq:P1P2}
\end{equation}

Each map has a one-dimensional kernel, and both induced families
$P_1\vartheta(E([s]))$ and $P_2\vartheta(E([s]))$ span $\R^4$. Hence both
quotients satisfy the dimensional counting part of
Proposition~\ref{prop:t3} with $N-4=1$.

For $P_2$, every four of the five vectors
\[
f_1,\ f_2,\ f_3,\ f_4,\ f_1+f_2+f_3+f_4
\]
are linearly independent. Hence $P_2\vartheta$ realizes the uniform matroid
$U_{4,5}$, and Eq.~\eqref{eq:closurefaithful} holds for every subset with
respect to $\cl_{\rm primary}$.

For $P_1$, by contrast, the fifth image is $f_1$. Therefore
\[
P_1\vartheta(e_5)
\in
\operatorname{span}P_1\vartheta(\{e_1\}),
\]
although $\{e_1,e_5\}$ is independent under $\cl_{\rm primary}$. The reverse
implication in Eq.~\eqref{eq:closurefaithful} therefore fails already on a
two-element subset.

Thus the two quotients have the same reduced rank but different operational
dependency structures.
\end{proof}

The example is written at rank four because that is the dimensional-reduction
application considered below, but the construction is rank independent.
Replacing $U_{4,5}$ by $U_{r,r+1}$ gives the same distinction for arbitrary
$r$: start from a free realization in $\R^{r+1}$ and project to $\R^r$ by
sending the final basis vector either to $f_1$ or to
$f_1+\cdots+f_r$.

The example shows that dimensional counting and dependency preservation are
independent requirements. A reduction may leave exactly four accessible
tangent directions while failing to realize the independently specified
rank-four operational closure, and the failure may already be visible on a
small subset of transitions. In the example above a single pair separates the
faithful and unfaithful quotients, making RB3 a genuine structural test rather
than a restatement of the desired rank.

\begin{remark}[Real representability is a genuine restriction]

Closure-faithfulness presupposes that the operational closure admits a real
linear realization at all, which is the content of B2-R. The uniform
structure used above is representable over $\R$, but real representability is
not automatic even at rank four. The Vámos matroid $V_8$, of rank four on
eight elements, is representable over no field \cite{oxley}. Hence
operational closure structures of rank four exist for which no real linear
realization is possible. B2-R excludes such structures by assumption; it is
therefore a genuine restriction on the domain of the bridge rather than a
technical convenience.

\end{remark}

\subsection{Kaluza--Klein circle reduction: a benign case}
\label{sec:kkcircle}

We first apply the operational reduction criterion to the simplest
Kaluza--Klein geometry, a construction that continues to arise in
higher-dimensional and quantum-gravity settings \cite{duffkk,liuMajid2025}.
Doing so requires an operational identification of the elementary transitions
that is fixed before the reduced rank is examined.
For a background with a transitive isometry group, we take the elementary
transitions to be the generators of that group acting on the local class
$[s]$. This identification is an explicit modelling assumption rather than a
consequence of the rank formalism; a different independently motivated choice
of elementary transitions could lead to a different operational structure.

Consider
\[
\Xmicro=\mathbb{R}^{1,3}\times S^1,
\]
with coordinates $x^\mu$, $\mu=0,\ldots,3$, and compact coordinate $y$ of
period $2\pi R$. Under the identification above,
\[
E([s])=\{e_0,\ldots,e_3,e_y\},
\]
with tangent realization
\[
\vartheta(e_\mu)=f_\mu,
\qquad
\vartheta(e_y)=f_y.
\]
The five generators are mutually independent, so the microscopic
accessibility structure is free of rank five and the realization is
closure-faithful.

Let the observer map $\Piobs$ be the geometric projection
\[
\Piobs\equiv\pi:
\mathbb{R}^{1,3}\times S^1\longrightarrow\mathbb{R}^{1,3}.
\]
Then
\[
\Ker d_x\Piobs=\operatorname{span}f_y,
\qquad
\rank d_x\Piobs=4,
\]
which realizes Corollary~\ref{cor:t2} with $N=5$ and $r=4$.

It is important to distinguish this geometric projection from the usual
Kaluza--Klein truncation of fields. The latter acts on field configurations,
not directly on the target manifold, and the four-dimensional zero-mode
sector generally retains fields descending from the compact geometry,
including the graviphoton and a radion-type scalar \cite{duffkk}. The present
example concerns the operational status of the compact target-space
direction within a low-energy observer sector whose admissible protocols do
not resolve momentum along $S^1$. It does not assert that every field-space
degree of freedom associated with the compact dimension disappears.

The operational question is therefore whether the compact generator $e_y$
remains part of the primary ground set. In the stipulated low-energy sector,
massive Kaluza--Klein excitations that resolve momentum along $y$ are excluded
from the admissible protocols. Consistently with this restriction, the
residual diffeomorphism associated with the compact coordinate appears in the
zero-mode description as the $U(1)$ gauge symmetry of the graviphoton
descending from $g_{\mu y}$; gauge-related configurations represent the same
physical state. In the pure circle case considered here, we therefore take
$e_y$ to leave $E_{\rm primary}([s])$ rather than to remain as an elementary
transition with a degenerate tangent image.

This situation is structurally simpler than the rank-four counterexample of
Proposition~\ref{prop:rankdoesnotdiscriminate}. The relevant reduction acts on
the accessibility structure as a matroid minor.

\begin{proposition}[Minors are benign]
\label{prop:minors}

Let the accessibility structure on $E([s])$ be matroidal, that is, let $\cl$
satisfy the closure axioms of a matroid, and let $\vartheta_x$ be
closure-faithful for it. Suppose the observer quotient acts on that structure
as a minor: either as deletion, with
\[
E_{\rm primary}([s])=E([s])\setminus\{e\}
\]
and $\cl_{\rm primary}$ the restriction of $\cl$, or as contraction by $e$,
with the realization composed with the quotient of $T_x\Xprimary$ by
$\langle\vartheta_x(e)\rangle$. Then the induced primary realization is
closure-faithful.

\end{proposition}

\begin{proof}

For deletion, both sides of Eq.~\eqref{eq:closurefaithful} involve only
subsets of $E_{\rm primary}([s])$ and are unchanged by removal of $e$.

For contraction, let $\vartheta'$ be the composition of $\vartheta_x$ with
the quotient map
\[
T_x\Xprimary
\longrightarrow
T_x\Xprimary/\langle\vartheta_x(e)\rangle.
\]
For $X\subseteq E([s])\setminus\{e\}$ and $a\notin X$,
\[
\vartheta'(a)\in\operatorname{span}\vartheta'(X)
\]
if and only if
\[
\vartheta_x(a)
\in
\operatorname{span}\!\bigl(\vartheta_x(X)\cup\{\vartheta_x(e)\}\bigr).
\]
By closure-faithfulness of $\vartheta_x$, this holds if and only if
\[
a\in\cl(X\cup\{e\}).
\]
Since the closure of the contraction satisfies
\[
\cl_{\rm primary}(X)=\cl(X\cup\{e\})\setminus\{e\},
\]
the operational and linear closure relations agree after contraction.
Therefore Eq.~\eqref{eq:closurefaithful} holds for $\vartheta'$.

\end{proof}

The matroid assumption is substantive. The record model of
Ref.~\cite{before} produces a free matroid, and B2-R already restricts the
bridge to operational closure structures admitting a real linear
representation. A non-matroidal accessibility structure falls outside
Proposition~\ref{prop:minors}; no corresponding preservation statement is
claimed for that case.

\begin{corollary}[The circle case is benign]
\label{cor:circle}

Under the stipulated geometric identification above, the differential quotient
associated with
\[
\pi:\mathbb R^{1,3}\times S^1\longrightarrow\mathbb R^{1,3}
\]
realizes contraction of $e_y$. Since the microscopic accessibility matroid is
free, $e_y$ is a coloop, and deletion and contraction induce the same free
matroid on the retained generators. Thus removing $e_y$ from the primary
ground set and quotienting its tangent direction are compatible descriptions
of the same retained accessibility structure.

By Proposition~\ref{prop:minors}, the induced primary realization remains
closure-faithful. Hence RB3 holds for this part of the reduction criterion.
Subject to PR0 and the remaining observer, regularity, and scale assumptions,
the circle reduction is therefore compatible with
Proposition~\ref{prop:t3}.

\end{corollary}

Thus the simplest circle compactification is a benign realization of the
criterion: under the stated operational identification, removal of the compact
generator preserves the free accessibility structure. The example establishes
compatibility with a standard reduction mechanism but does not yet discriminate
between viable reductions.

A stronger test arises when information from the compact sector survives in
the low-energy primary sector, for example through Wilson-line or
Aharonov--Bohm observables. In that case the reduction is no longer represented
solely by deletion or contraction, and the surviving dependency structure must
be tested explicitly.

\subsection{Compactness does not reduce operational rank in the free commuting case}

The benign outcome of the circle example raises a more general question.
Does compactness itself create operational dependence among otherwise
independent directions? For the simplest commuting generation pattern the
answer is no. Compactness identifies points along the flow of an individual
generator, but it does not by itself make that generator reproducible from
the others.

\begin{proposition}[Compactness is invisible for freely generated commuting
structures]
\label{prop:free}

Suppose the local monoid $\mathcal P([s])$ is generated by finitely many
commuting one-parameter subgroups whose generators are linearly independent in
the corresponding Lie algebra. Assume further that no non-trivial finite
composition involving one of these subgroups is exactly reproducible by the
others, so that the only relations are internal to the individual factors.
Let
\[
E([s])=\{e_1,\ldots,e_m\}
\]
be the associated elementary transitions. Then the accessibility matroid on
$E([s])$ is free of rank $m$, irrespective of which of the generating
subgroups are compact.

\end{proposition}

\begin{proof}

A transition belongs to $\cl(X)$ precisely when it can be reproduced by a
finite composition of procedures associated with elements of $X$. Since the
generating one-parameter subgroups commute, every such composition remains
inside the subgroup generated by those factors.

Let $a\notin X$. By hypothesis, the flow generated by $e_a$ is not
reproducible by a finite composition of the factors generated by $X$.
Therefore
\[
e_a\notin\cl(X).
\]
Hence no element outside a subset belongs to its operational closure, so every
subset of $E([s])$ is independent and the accessibility matroid is free.

If a generating subgroup is compact, compactness identifies a finite period
along its own flow. This is an internal relation within one factor and does
not express that generator as a composition of the remaining generators.
The additional hypothesis excludes precisely the global relations between
different factors that could invalidate this conclusion.

\end{proof}

The additional independence assumption is essential. Linear independence of
the infinitesimal generators alone does not exclude discrete global relations
between distinct factors when compact quotients are present. The hypothesis
is satisfied, in particular, when the subgroup generated by the elementary
flows is a direct product of the corresponding one-parameter factors.

Proposition~\ref{prop:free} shows that compactness alone does not lower the
operational rank in this freely generated commuting class. A compact generator
remains operationally independent as long as it belongs to the primary ground
set; reducing the rank therefore requires an additional physical operation,
such as removing or integrating out that direction at a specified scale.

This also explains the benign circle example of
Section~\ref{sec:kkcircle}: deletion or contraction of the compact generator
preserves the free accessibility structure. The situation changes when
compact-sector information survives the reduction or when operational
generation ceases to be free. The latter possibility occurs naturally in
non-Abelian settings, where loop observables may satisfy non-trivial
compositional relations \cite{mandelstam}.

\subsection{Surviving Wilson-line directions}
\label{sec:wilsonlines}

A more demanding case occurs when information associated with the compact
direction survives in the low-energy theory. Consider a gauge field $A_M$,
with $M=\mu,y$, on
\[
\mathbb{R}^{1,3}\times S^1.
\]
The zero-mode reduction retains a four-dimensional gauge field $A_\mu$
together with a scalar $A_y$. The gauge-invariant content of the compact
component is encoded in the holonomy
\begin{equation}
 \theta=g\oint A_y\,dy\pmod{2\pi},
 \label{eq:holonomy}
\end{equation}
which can remain observable in four dimensions through an
Aharonov--Bohm phase for a charged state \cite{hosotani2015}, even though
direct translation along the compact coordinate is no longer part of the
low-energy observer sector.

This example requires an extension of the operational identification used in
Section~\ref{sec:kkcircle}. There the elementary transitions were identified
with isometry generators. A change of the Wilson-line modulus is instead a
motion in gauge-field configuration space, not an isometry of
$\mathbb{R}^{1,3}\times S^1$. We therefore adopt the following additional
modelling assumption: parallel transport around the compact cycle is treated
as a composite operational procedure, and the resulting holonomy degree of
freedom is admitted into the primary sector as an elementary transition.
This identification is fixed before any rank is evaluated.

Such a change of ground set is not a deletion or contraction, so
Proposition~\ref{prop:minors} does not apply. In general, promotion of a
composite transition could lead to the type of mismatch exhibited in
Proposition~\ref{prop:rankdoesnotdiscriminate}: the promoted transition might
be operationally independent even though its tangent representative lies in
the span of already retained directions.

For the simplest commuting realization, however,
Proposition~\ref{prop:free} gives a sharper result. If the four non-compact
translations together with the holonomy shift are freely generated commuting
elementary transitions, then the primary accessibility matroid is the free
matroid $U_{5,5}$. Consequently,
\[
\rank\Apr([s])=5.
\]
Compactness of the extra direction does not restore rank four; as long as the
holonomy shift remains an independently admissible primary transition, it
provides a genuine fifth operational direction.

The physical fate of that direction is nevertheless scale and
model dependent. Stabilizing the Wilson-line modulus is not equivalent to
removing it. If
\[
\theta=\theta_0+\delta\theta,
\qquad
V''(\theta_0)>0,
\]
then $\delta\theta$ is a massive four-dimensional scalar and remains a local
tangent direction of the configuration manifold. Four physically distinct
situations should therefore be separated.

\begin{enumerate}[label=(\alph*)]

\item \textbf{Flat or approximately flat holonomy direction.}
If shifts of $\theta$ are admissible primary procedures and the corresponding
outcome probabilities vary smoothly with $\theta$, then the holonomy supplies
an additional locally resolvable direction. In the freely generated
commuting realization above, its operational rank contribution is not merely
possible but fixed by Proposition~\ref{prop:free}: the primary operational
rank is five.

If the smooth and observer bridge assumptions of
Sections~\ref{sec:smooth-realization} and \ref{sec:observer} also hold for
this independently specified primary sector, then
Corollary~\ref{cor:t1c} transfers this fifth operational direction to the
Fisher realization. Such a sector is incompatible with the rank-four
specialization T1a under the assumptions required for that specialization.
If the bridge fails instead, no inference from the measured Fisher rank to
the operational rank is available.

\item \textbf{Massive but retained holonomy field.}
A potential can lift the modulus without eliminating the corresponding
direction. The fluctuation $\delta\theta$ remains part of the local
configuration manifold, while its operational accessibility depends on the
energy scale and on the admitted protocol family. This is primarily an RB4
issue: the same microscopic background can define different primary sectors
at different observational scales.

\item \textbf{Integrated out below the mass scale.}
At energies well below its mass, the holonomy field may be explicitly
integrated out of the effective description. Only in this case is the
direction absent from the effective primary sector by construction. The
coarse-graining must be specified together with its physical scale, and
Proposition~\ref{prop:coarse} then guarantees that such a coarse-graining
cannot increase the operational rank.

\item \textbf{Discrete or topologically distinct sectors.}
The surviving information may be a discrete label rather than a continuous
field direction. Distinct Wilson-line vacua or topological sectors can remain
operationally distinguishable by global protocols even when no local tangent
direction connects them. A locally constant label has vanishing differential,
and therefore contributes neither to $d\Piobs$ nor to $d\Phi_M$ nor to any
Fisher matrix. Such distinguishability is genuine but global, and lies
outside the differential rank construction of
Sections~\ref{sec:smooth} and \ref{sec:observer}.

\end{enumerate}

Which regime is realized depends on the compactification data. Wilson-line
moduli can acquire quantum effective potentials of Hosotani type
\cite{hosotani2015,hosotani}, while their accessibility may also depend on
fluxes, brane effects, instantons, boundary conditions, matter content, and
supersymmetry. The operational criterion must therefore be applied relative
to a specified primary observer sector and physical scale.

The sharpest case is a surviving freely generated holonomy direction. Then
compactification alone does not reduce the primary operational rank to four:
the holonomy contributes a fifth direction, and an additional physical
mechanism is required to remove, identify, or render it operationally
inaccessible. Massive, integrated-out, and discrete sectors instead show how
this conclusion depends on scale and on whether the surviving information is
represented by a local tangent direction.

The circle and Wilson-line examples therefore probe distinct regimes. In the
former, the compact generator leaves the primary ground set and the reduction
is a benign minor operation. In the latter, compact-sector information
survives and may contribute an additional operational direction. The remaining
case is one in which the surviving operational structure is not freely
generated.

\subsection{Non-Abelian generation can obstruct closure-faithfulness}
\label{sec:nonabelian}

The situation changes when operational generation is non-Abelian. A
transition may then be generated compositionally by other procedures while
its tangent representative remains outside their linear span, producing a
direct obstruction to RB3.

\begin{proposition}[Compositional obstruction]
\label{prop:compositional}

Let $S\subseteq E([s])$ and $e_c\in E([s])$ be such that $e_c$ is reproduced
by a finite composition of procedures associated with $S$, so that
$e_c\in\cl(S)$. If the smooth realization satisfies
\[
\vartheta_x(e_c)\notin\operatorname{span}\vartheta_x(S),
\]
then the realization is not closure-faithful at $x$ and RB3 fails.

\end{proposition}

\begin{proof}

The operational dependence relation gives
$e_c\in\cl(S)$, whereas the corresponding linear dependence relation fails by
hypothesis. Hence the equivalence required by
Eq.~\eqref{eq:closurefaithful} is violated.

\end{proof}

The proposition becomes physically relevant when the compositional relation is
generated by an exact group-theoretic identity rather than inferred from an
infinitesimal approximation. Two standard examples exhibit precisely this
structure.

\begin{corollary}[$\mathfrak{su}(2)$ generation]
\label{cor:su2}

Let $E([s])$ contain elementary transitions $e_X,e_Y,e_Z$ associated with the
one-parameter flows generated by $X,Y,Z$, realized by left-invariant vector
fields on $SU(2)$, with
\[
[X,Y]=Z.
\]
Since $X$ and $Y$ generate $\mathfrak{su}(2)$ as a Lie algebra, the connected
subgroup generated by the one-parameter subgroups
$\{\exp(tX)\}$ and $\{\exp(tY)\}$ has Lie algebra
$\mathfrak{su}(2)$ \cite{hall2015}. Because $SU(2)$ is connected, this
subgroup is all of $SU(2)$. Hence $\exp(tZ)$ belongs to the subgroup generated
by the flows of $X$ and $Y$ and can be expressed as a finite composition of
such flows. Therefore
\[
e_Z\in\cl(\{e_X,e_Y\}).
\]

At every point, however, $Z(x)$ is linearly independent of $X(x)$ and
$Y(x)$. Thus
\[
\vartheta_x(e_Z)
\notin
\operatorname{span}\{\vartheta_x(e_X),\vartheta_x(e_Y)\},
\]
so RB3 fails by Proposition~\ref{prop:compositional}.

At the same time,
\[
\Dprimary=\operatorname{span}\{X,Y,Z\}
\]
has constant rank three and is involutive. Hence RB1 and RB2 hold while RB3
fails.

\end{corollary}

\begin{remark}[The Heisenberg example revisited]

The same obstruction appears in the three-dimensional Heisenberg group. The
exact identity
\begin{equation}
 \exp(aX)\exp(bY)\exp(-aX)\exp(-bY)=\exp(abZ)
\end{equation}
implies
\[
e_Z\in\cl(\{e_X,e_Y\}),
\]
while $\vartheta(e_Z)$ lies outside the linear span of the images of $e_X$ and
$e_Y$. Thus the example of Section~\ref{sec:lie}, introduced there to show
that a count of generators does not determine an orbit dimension, is also an
explicit realization satisfying RB1 and RB2 but violating RB3
\cite{hall2015}.

\end{remark}

These examples separate integrability from closure-faithfulness. RB2 concerns
whether the chosen distribution is closed under Lie brackets and therefore
integrable as a smooth distribution. RB3 instead asks whether operational
dependence under finite composition agrees with linear dependence of the
tangent representatives. The two conditions are logically independent.

Corollary~\ref{cor:su2} gives one direction of this independence: RB1 and RB2
hold while RB3 fails. The converse can be exhibited explicitly. On $\R^3$,
take
\[
E([s])=\{e_X,e_Y\}
\]
with free operational closure and realize the two elementary transitions by
\begin{equation}
 X=\partial_x+y\partial_z,
 \qquad
 Y=\partial_y,
 \qquad
 [X,Y]=-\partial_z.
 \label{eq:noninvolutive}
\end{equation}
The fields $X$ and $Y$ are linearly independent everywhere, and neither
elementary transition is reproducible from the other alone. Hence
Eq.~\eqref{eq:closurefaithful} holds on every subset, so RB3 is satisfied, as
is the constant-rank condition RB1.

However,
\[
aX+bY
=
a\partial_x+b\partial_y+ay\partial_z
\]
cannot equal $-\partial_z$ for any functions $a$ and $b$. The bracket
therefore leaves the span of the distribution and RB2 fails. Together with
Corollary~\ref{cor:su2}, this establishes that neither RB2 nor RB3 implies the
other.

For dimensional reduction this differs qualitatively from the freely
generated commuting case. A non-Abelian realization may possess a regular
involutive smooth distribution and still fail the operational-to-smooth
bridge because finite compositional dependence does not agree with linear
tangent dependence. The obstruction therefore concerns the identification of
operational and smooth directions, not the existence of a smooth reduced
geometry.

This is not a general no-go theorem for non-Abelian compactifications or
Wilson-line sectors. Whether the obstruction occurs depends on the chosen
primary ground set, the finite compositional relations in its operational
closure, and the tangent realization of those transitions.
Corollary~\ref{cor:su2} establishes that such an obstruction can occur, not
that it must occur in every non-Abelian realization.

The examples of Sections~\ref{sec:kkcircle}--\ref{sec:nonabelian} therefore
separate three regimes of dimensional reduction. In freely generated
commuting reductions, removing compact directions by deletion or contraction
is structurally benign. If compact-sector observables such as Wilson-line
moduli survive, they can contribute additional primary operational
directions. In non-Abelian compositional structures, the rank-preserving
bridge can fail altogether because finite operational generation and linear
tangent dependence need not coincide.

\section{Quantum-hardware negative control for Fisher-rank completeness}
\label{sec:calibration}

The observer-to-Fisher step of the rank-preserving bridge requires a
completeness assumption: a restricted measurement protocol can fail to resolve
directions that remain physically present in the parametrized family. The
purpose of this section is to isolate this failure mode experimentally. We
therefore construct a controlled quantum-circuit negative control in which the
underlying local Fisher rank is known independently, while deliberately
restricted measurements return a smaller rank.

This experiment is not intended to realize the primary operational structure
of Sections~\ref{sec:operational-rank}--\ref{sec:smooth-realization}, nor to
test the rank-four dimensional-reduction hypothesis. It tests only the
measurement layer of the logical chain: whether a stable restricted Fisher
rank can be identified with the full locally resolvable rank without an
independent completeness argument.

Consider a three-qubit six-parameter circuit,
\begin{equation}
\begin{aligned}
|000\rangle
&\to R_x^{(0)}(q_0)R_y^{(1)}(q_1)\\
&\to R_z^{(2)}(q_2)
   \to CX_{0\to1}
   \to CX_{1\to2}\\
&\to R_y^{(0)}(q_3)R_z^{(1)}(q_4)\\
&\to R_x^{(2)}(q_5)
   \to Z\text{-measurement}.
\end{aligned}
\label{eq:circuit}
\end{equation}

The parameter $q_2$ acts through $R_z^{(2)}$ on qubit 2 while that qubit is
still in $|0\rangle$, so at this stage it contributes only an overall phase
to the full state. Consequently $q_2$ is an exact null direction of the
original parametrized family for all physical measurements. The parameter
$q_4$, by contrast, changes a relative phase of the state and is
therefore physically non-null, but the final computational-basis measurement
is insensitive to that phase at the chosen circuit output. Consequently the
restricted $Z$ protocol has
\begin{equation}
 \rank F_Z\le4
 \label{eq:FZ4}
\end{equation}
by construction.

These relations, and the participation ratios they carry, can be computed
exactly. Writing the derivatives of the state analytically and diagonalizing
at a representative regular parameter point
\begin{equation}
 q=(0.37,0.83,1.21,0.59,1.07,1.43),
\end{equation}
one obtains the spectra of Table~\ref{tab:fisher}, in which $\deff$ denotes
the participation ratio of Section~\ref{sec:scale} evaluated on the Fisher
matrix at hand,
\begin{equation}
 \deff=\frac{(\operatorname{tr}F)^2}{\operatorname{tr}(F^2)}
 =\frac{(\sum_i\lambda_i)^2}{\sum_i\lambda_i^2}.
 \label{eq:deff}
\end{equation}

For reproducibility, the tomographically complete (IC) Pauli protocol used here
contains all 27 settings in $\{X,Y,Z\}^{\otimes 3}$ with equal weight
$1/27$, so
\begin{equation}
 F_{\rm IC}=\frac1{27}\sum_{s\in\{X,Y,Z\}^{\otimes3}}F_s.
 \label{eq:FICweight}
\end{equation}
The numerical script uses exact analytic state derivatives, support threshold
$p_i>10^{-12}$ in the classical Fisher sum, and counts an eigenvalue as
non-zero when it exceeds
\begin{equation}
 10^{-8}\max\{1,\lambda_{\max}\}.
 \label{eq:ranktol}
\end{equation}
The genericity check samples $200$ points from the uniform box
$q_i\in[0.2,2.9]$ with NumPy random seed \texttt{20260815}. These conventions
affect the reported numerical spectra and reproducibility details but not the
analytic rank statements in
Eqs.~\eqref{eq:ranks-original}--\eqref{eq:ranks-corrected}.

\begin{table*}[htbp]
\centering
\caption{
    Fisher ranks and participation ratios for the six-parameter circuit
    \eqref{eq:circuit} at a representative regular parameter point, for the
    restricted computational-basis protocol, for a tomographically complete
    product-Pauli protocol of twenty-seven equally weighted settings, and for the
    quantum Fisher information. The corrected family replaces the first-layer
    $R_z^{(2)}(q_2)$ by $R_y^{(2)}(q_2)$. Ranks are stable over two hundred
    randomly sampled parameter points; the participation ratios are not, and their
    observed ranges are given in the last column
}
\label{tab:fisher}
\renewcommand{\arraystretch}{1.2}
\begin{tabularx}{\textwidth}{>{\raggedright\arraybackslash}p{0.25\textwidth} X
c c c}
\toprule
\textbf{Family and protocol} & \textbf{Fisher eigenvalues} & \textbf{rank} &
$\boldsymbol{d_{\rm eff}}$ & $\boldsymbol{d_{\rm eff}}$ \textbf{range}\\
\midrule
Original, $Z$ basis & 1.000, 0.922, 0.876, 0.374 & 4 & 3.649 & 1.24--3.97\\
Original, IC Pauli & 0.695, 0.558, 0.488, 0.465, 0.171 & 5 & 4.421 &
3.72--4.84\\
Original, quantum Fisher & 1.000, 1.000, 1.000, 1.000, 0.533 & 5 & 4.796 &
4.04--5.00\\
Corrected, $Z$ basis & 0.922, 0.876, 0.374, 0.140 & 4 & 3.007 & 1.55--3.92\\
Corrected, IC Pauli & 0.664, 0.589, 0.544, 0.464, 0.246, 0.043 & 6 & 4.777 &
4.03--5.68\\
Corrected, quantum Fisher & 1.000, 1.000, 1.000, 1.000, 0.533, 0.125 & 6 &
5.045 & 4.08--5.99\\
\bottomrule
\end{tabularx}
\end{table*}

The rank entries of Table~\ref{tab:fisher} are the relations quoted below. The
tomographically complete product-Pauli protocol reveals the $q_4$ direction,
giving, at the reference point and throughout the sampled regular region,
\begin{equation}
 \rank F_Z=4,
 \qquad
 \rank F_{\rm IC}=5,
 \qquad
 \rank F_Q=5
 \label{eq:ranks-original}
\end{equation}
for the original family. 
Replacing the initial $R_z^{(2)}(q_2)$ with $R_y^{(2)}(q_2)$ removes the
global-phase redundancy, and the corrected family attains
\begin{equation}
 \rank F_{\rm IC}=\rank F_Q=6
 \label{eq:ranks-corrected}
\end{equation}
while a restricted $Z$ measurement may still return rank four.

The two families therefore provide a paired negative control for rank
interpretation: the same restricted Fisher rank four is compatible with two
different larger locally resolvable ranks. The issue is measurement
completeness, not the numerical value four itself.

The relation between the two circuit families and the restricted and revealing
measurement protocols is summarized schematically in
Fig.~\ref{fig:circuit-protocol}. The two families have the same three-qubit
architecture and differ only by the replacement $R_z^{(2)}(q_2)\to
R_y^{(2)}(q_2)$ in the first local-rotation layer. Under the restricted $ZZZ$
protocol both families exhibit rank four, whereas adding the revealing $YXX$
setting exposes a fifth direction in the original family and a sixth direction
in the corrected family.

\begin{figure*}[t!]
    \centering
    \includegraphics[width=\textwidth]{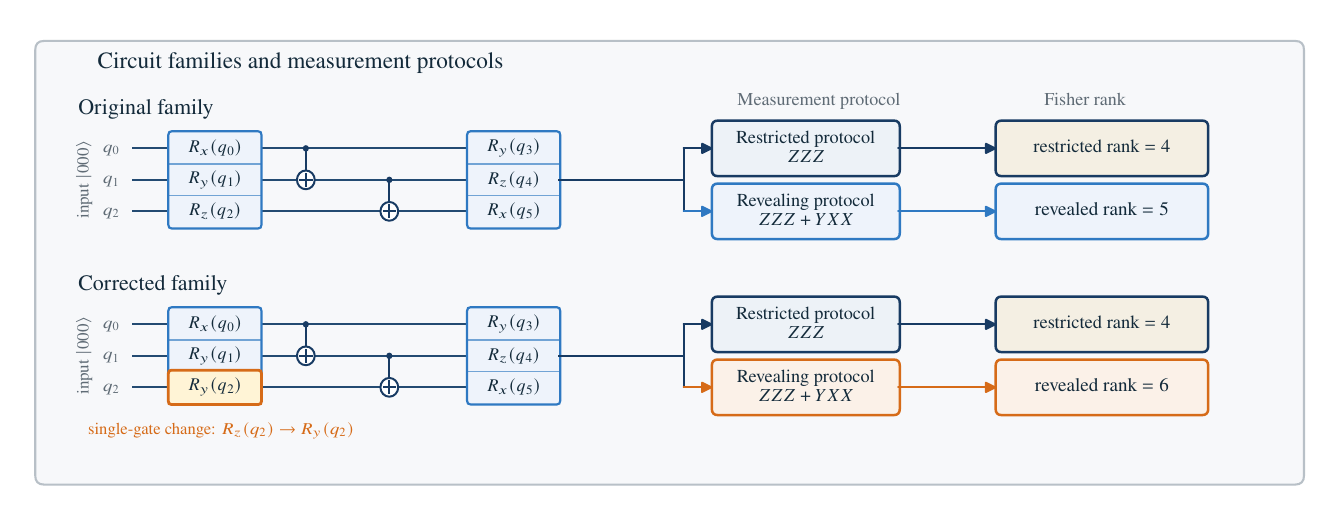}
    \caption{
    Circuit families and measurement protocols used in the paired
    Fisher-rank negative control. The corrected family differs from the
    original family only by the single-gate replacement
    $R_z^{(2)}(q_2)\to R_y^{(2)}(q_2)$. The restricted protocol
    $\{ZZZ\}$ returns rank four for both families, while the revealing
    protocol $\{ZZZ,YXX\}$ exposes rank five for the original family
    and rank six for the corrected family
    }
    \label{fig:circuit-protocol}
\end{figure*}

The participation ratios provide an additional illustration of the same
incompleteness problem. Every row of Table~\ref{tab:fisher} satisfies the bound of
Proposition~\ref{prop:deffrank}, as it must. But $\deff$ is a strictly weaker
diagnostic than the rank: over the sampled parameter points the restricted
protocol, whose rank is four everywhere, returns participation ratios anywhere
between $1.24$ and $3.97$, while the tomographically complete product-Pauli
protocol, whose rank is five, returns values from $3.72$ upwards. The two
ranges overlap. A measured participation ratio therefore does not determine
the rank, does not distinguish a complete protocol from an incomplete one, and
cannot be read as an effective dimension of the underlying sector. An observed
value near unity is compatible with rank four; a value near four is compatible
with rank six.

\subsection{Pre-specified revealing measurement protocol}

The full 27-setting product-Pauli protocol establishes the numerical reference
rank, but informational completeness is not required for the negative-control
statement. To disprove the inference from restricted rank to full locally
resolvable rank, it is sufficient to pre-specify one admissible extension of
the restricted protocol that reveals an additional independent direction.
The reduced hardware workload was therefore designed and frozen numerically
before any hardware data were acquired.

The resulting rank pattern is summarized in Fig.~\ref{fig:rank-summary}. The 
important point is that the restricted protocol assigns rank four to both
circuit families, whereas the pre-specified revealing protocol distinguishes
their larger locally resolvable ranks.

\begin{figure*}[htbp]
    \centering
    \includegraphics[width=\textwidth]{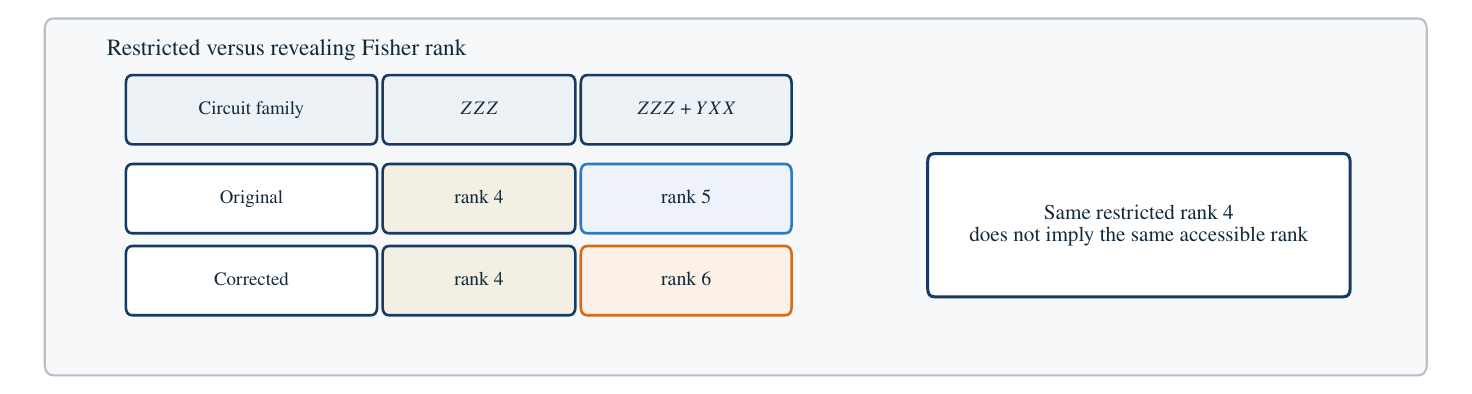}
    \caption{
    Restricted versus revealing Fisher rank for the two circuit
    families. The same restricted rank four is compatible with
    different larger accessible ranks: five for the original family
    and six for the corrected family
    }
    \label{fig:rank-summary}
\end{figure*}

Starting from the original family, we searched over the 26 Pauli-product
settings additional to $ZZZ$. At the reference point used above, exhaustive
search showed that a single additional setting is sufficient: the two-setting
protocol
\begin{equation}
 \mathcal M_{\rm reveal}=\{ZZZ,YXX\}
 \label{eq:revealprotocol}
\end{equation}
has Fisher rank five, whereas the restricted protocol
\begin{equation}
 \mathcal M_{\rm restricted}=\{ZZZ\}
\end{equation}
has rank four. Thus informational completeness is not required for the
negative-control statement.

The two settings are combined with equal, parameter-independent weights
$1/2$, so that
\begin{equation}
 F_{\rm reveal}
 =
 \frac12 F_{ZZZ}
 +
 \frac12 F_{YXX}.
 \label{eq:revealweights}
\end{equation}

To maximize the separation of the fifth Fisher direction from the null
spectrum before hardware acquisition, the working point was chosen by a
numerical search over the six gate-angle coordinates. The selected point was
\begin{equation}
\begin{aligned}
q_* = (&1.550560,\,
        2.890920,\,
        1.961725,\\
      &1.785682,\,
        2.259350,\,
        1.814126).
\end{aligned}
\label{eq:qstar}
\end{equation}
For the exact two-setting Fisher matrix at this point,
\begin{equation}
 \frac{\lambda_5}{\lambda_1}\simeq0.878,
\end{equation}
so the additional direction is not a weak near-null mode.

All choices in Eqs.~\eqref{eq:revealprotocol}--\eqref{eq:qstar} were fixed
before the hardware run. The same design was then tested locally with
finite-shot parameter-shift estimation and with a device-noise simulation.
Exact derivatives were evaluated by the parameter-shift rule,
\begin{equation}
 \frac{\partial p_i}{\partial q_a}
 =
 \frac{1}{2}
 \left[
 p_i\!\left(q_a+\frac{\pi}{2}\right)
 -
 p_i\!\left(q_a-\frac{\pi}{2}\right)
 \right].
 \label{eq:paramshift}
\end{equation}

The corresponding estimator chain is shown in
Fig.~\ref{fig:estimator-pipeline}. The circuit parameters determine the state
family $U(q)$, the chosen measurement protocol produces outcome probabilities
$p(x\mid q)$, parameter-shift derivatives form the local Jacobian $J$, and the
Fisher matrix $F=J^{\mathsf T}WJ$ is finally diagonalized to obtain the
spectrum used in the rank test.

\begin{figure*}[htbp]
    \centering
    \includegraphics[width=\textwidth]{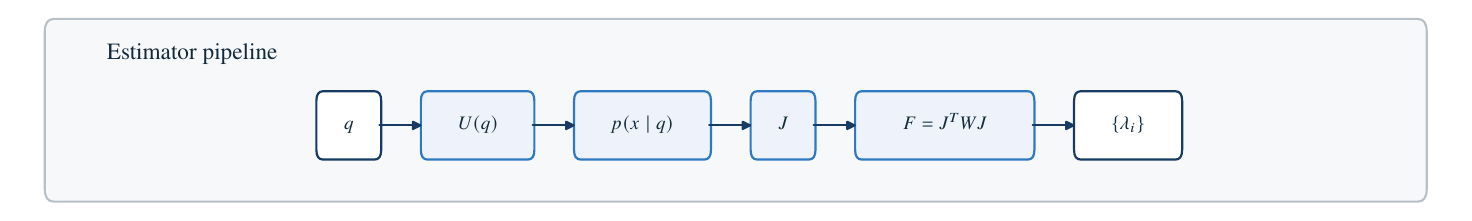}
    \caption{
    Estimator pipeline used in the circuit calibration. The parameter
    vector $q$ defines the circuit family $U(q)$; the measurement
    protocol determines $p(x\mid q)$; local derivatives form the
    Jacobian $J$; and the Fisher matrix
    $F=J^{\mathsf T}WJ$ yields the eigenvalue spectrum
    $\{\lambda_i\}$ used for rank inference
    }
    \label{fig:estimator-pipeline}
\end{figure*}

Because finite sampling generically lifts exact zero eigenvalues, experimental
rank was not assigned by a fixed numerical eigenvalue cutoff. Instead, the
fifth eigenvalue of the revealing protocol was compared with a null
distribution constructed from two known finite-shot null channels: the fifth
eigenvalue of the rank-four $ZZZ$ protocol and the sixth eigenvalue of the
rank-five revealing protocol. At $\alpha=0.01$, the conservative null threshold 
was defined as
\begin{equation}
\begin{aligned}
\lambda_{\rm null}^{(1-\alpha)}
={}&
\max\!\Bigl\{
Q_{1-\alpha}\!\left[\lambda_5(F_{ZZZ})\right],
\\
&\qquad
Q_{1-\alpha}\!\left[\lambda_6(F_{\rm reveal})\right]
\Bigr\}.
\end{aligned}
\label{eq:nullthreshold}
\end{equation}
where $Q_{1-\alpha}$ denotes the corresponding bootstrap quantile. The
pre-declared hardware criterion required the revealing fifth eigenvalue to
exceed this threshold robustly; no post-hardware optimization of $q_*$, the
measurement settings, or the threshold was performed.

\subsection{IBM Quantum hardware result}

The negative control was executed independently on three
156-qubit IBM Quantum processors. The selected backends were
\texttt{ibm\_marrakesh}, \texttt{ibm\_fez}, and \texttt{ibm\_kingston}. 
The measurement design, working point $q_*$, parameter-shift rule, 
number of shots, and statistical decision criterion were fixed before 
any hardware data were acquired and were kept unchanged across all three 
processors.

Calibration data were inspected immediately before each submission. Connected
physical-qubit chains were selected independently for each backend:
\begin{equation}
\begin{aligned}
 \texttt{ibm\_marrakesh}:&\quad (13,14,15),\\
 \texttt{ibm\_fez}:&\quad (97,107,108),\\
 \texttt{ibm\_kingston}:&\quad (140,141,142).
\end{aligned}
\label{eq:hardware-qubits}
\end{equation}
For every backend the compiled circuits used only the selected three physical
qubits and contained exactly two native two-qubit gates, with no routing
SWAPs.

For each of the two measurement settings, the central point $q_*$ and the
twelve parameter-shift configurations $q_a\pm\pi/2$, $a=0,\ldots,5$, were
executed. Each backend therefore received the same workload,
\begin{equation}
 2(1+2\times6)=26
\end{equation}
circuits, each sampled with $1024$ shots, for a total of $26624$ shot
executions per processor. All 26 circuits for a given processor were submitted
in a single IBM Runtime Sampler job. The reported quantum processing unit (QPU) 
charge time was $9\,\mathrm{s}$ on each processor, giving $27\,\mathrm{s}$ in total for 
the three hardware realizations.

The revealing protocol reproduced the same qualitative and quantitative rank
separation on all three processors. The observed fifth Fisher eigenvalues were
\begin{equation}
\begin{aligned}
 \lambda_{5,\rm Marrakesh}^{\rm reveal}&=0.38110,\\
 \lambda_{5,\rm Fez}^{\rm reveal}&=0.40081,\\
 \lambda_{5,\rm Kingston}^{\rm reveal}&=0.41242.
\end{aligned}
\label{eq:hardwarelambda5}
\end{equation}
The corresponding conservative $99\%$ finite-shot null thresholds defined in
Eq.~\eqref{eq:nullthreshold} were
\begin{equation}
\begin{aligned}
 \lambda_{{\rm null},\rm Marrakesh}^{99\%}&=0.00893,\\
 \lambda_{{\rm null},\rm Fez}^{99\%}&=0.00912,\\
 \lambda_{{\rm null},\rm Kingston}^{99\%}&=0.01396.
\end{aligned}
\label{eq:hardwarenull}
\end{equation}
Thus the observed signal-to-null-threshold ratios were
\begin{equation}
\begin{aligned}
 R_{\rm Marrakesh}&=42.7,\\
 R_{\rm Fez}&=43.9,\\
 R_{\rm Kingston}&=29.5.
\end{aligned}
\label{eq:hardwareratio}
\end{equation}

A parametric bootstrap with $10^4$ resamples conditioned on the observed
hardware probabilities gave a bootstrap-estimated detection power of $1.000$
on all three processors under the pre-declared criterion. In each case no null
bootstrap replicate exceeded the observed statistic, so the empirical
null-tail probability reached the resolution limit of the $10^4$-sample
bootstrap. The fifth direction therefore remained strongly separated from the
empirical finite-shot null spectrum on three independent hardware
realizations.

The full revealing Fisher spectra also showed the same structure across
\begin{equation}
\begin{aligned}
\lambda(F_{\rm reveal})\simeq
(&0.5285,\,
  0.4869,\,
  0.4744,\\
 &0.4530,\,
  0.4008,\,
  0.00134).
\end{aligned}
\end{equation}
\begin{equation}
\begin{aligned}
\lambda(F_{\rm reveal})\simeq
(&0.5368,\,
  0.4655,\,
  0.4486,\\
 &0.4421,\,
  0.4124,\,
  0.00505).
\end{aligned}
\end{equation}
Together with the Marrakesh result, these spectra exhibit five strong
sensitivity directions and a sixth direction remaining near the null floor.

The known-null/hidden-direction diagnostic is likewise reproduced across
processors. For \texttt{ibm\_fez},
\begin{equation}
 F_{q_2q_2}\simeq0.00233,
 \qquad
 F_{q_4q_4}\simeq0.48775,
\end{equation}
with a bootstrap median ratio $F_{q_4q_4}/F_{q_2q_2}\simeq86.1$. For
\texttt{ibm\_kingston},
\begin{equation}
 F_{q_2q_2}\simeq0.00594,
 \qquad
 F_{q_4q_4}\simeq0.46015,
\end{equation}
with the corresponding median ratio $\simeq50.4$. Thus the direction revealed
by the additional $YXX$ setting remains far above the known global-phase null
direction despite differences in backend calibration.

For this calibration, let $\mathcal M_{\rm cal}$ denote the pre-specified
measurement family containing the restricted and revealing protocols. The
hardware result establishes
\begin{equation}
\begin{aligned}
\rank F_{ZZZ}=4
&\quad\text{does not imply}\\
&\sup_{M\in\mathcal M_{\rm cal}}
\rank F^{(M)}=4.
\end{aligned}
\label{eq:hardware-negative}
\end{equation}

Adding the single pre-specified $YXX$ setting reveals a fifth sensitivity
direction far above the finite-shot null spectrum on all three processors.
Thus the apparent rank reduction produced by the restricted protocol is a
measurement-blindness effect rather than a property of the full locally
resolvable circuit family.

The scope of this statement is deliberately limited. The experiment concerns
the Fisher geometry of a controlled parametrized quantum family. It does not
identify a primary operational sector, establish PI or closure-faithfulness,
or test Hypothesis H. Its role in the present work is therefore a
hardware-validated negative control for the Fisher-completeness step of the
rank-preserving bridge, not evidence for fundamental four-dimensionality.

The cross-backend results are summarized in Table~\ref{tab:hardware}.

\begin{table*}[htbp]
\centering
\caption{
    Cross-backend IBM Quantum realization of the pre-specified
    negative-control protocol. The same working point, measurement settings,
    parameter-shift configurations, shot count, and statistical decision rule were
    used on all three processors
}
\label{tab:hardware}
\renewcommand{\arraystretch}{1.2}
\begin{tabular}{lccc}
\toprule
\textbf{Quantity}
& \textbf{Marrakesh}
& \textbf{Fez}
& \textbf{Kingston}\\
\midrule
Physical qubits
& $(13,14,15)$
& $(97,107,108)$
& $(140,141,142)$\\
Circuits
& $26$ & $26$ & $26$\\
Shots per circuit
& $1024$ & $1024$ & $1024$\\
Shot executions
& $26624$ & $26624$ & $26624$\\
$\lambda_5^{\rm reveal}$
& $0.38110$ & $0.40081$ & $0.41242$\\
$99\%$ null threshold
& $0.00893$ & $0.00912$ & $0.01396$\\
Signal/null ratio
& $42.7$ & $43.9$ & $29.5$\\
Bootstrap trials
& $10^4$ & $10^4$ & $10^4$\\
Detection power
& $1.000$ & $1.000$ & $1.000$\\
Null-tail result
& No exceedance & No exceedance & No exceedance\\
QPU charge time
& $9\,\mathrm{s}$ & $9\,\mathrm{s}$ & $9\,\mathrm{s}$\\
\bottomrule
\end{tabular}
\end{table*}

\subsection{Paired hardware control: the same restricted rank hides ranks five and six}

The first hardware experiment establishes that restricted rank four can conceal
a fifth locally resolvable direction. A stronger negative control asks whether
the same restricted rank can be compatible with \emph{different} larger
ranks. To test this, we compare two closely related six-parameter circuit
families under exactly the same restricted and revealing measurement
protocols.

The original family is the circuit of Eq.~\eqref{eq:circuit}. In the corrected
family the first-layer gate
\begin{equation}
 R_z^{(2)}(q_2)
\end{equation}
is replaced by
\begin{equation}
 R_y^{(2)}(q_2).
\end{equation}
The replacement removes the exact global-phase redundancy of $q_2$ while
leaving the six-parameter architecture otherwise unchanged. The exact
benchmark of Table~\ref{tab:fisher} already predicts the qualitative
distinction:
\begin{equation}
\begin{aligned}
 \text{original:}\qquad
 &\rank F_{ZZZ}=4,
 &\qquad \rank F_{\rm IC}=5,\\
 \text{corrected:}\qquad
 &\rank F_{ZZZ}=4,
 &\qquad \rank F_{\rm IC}=6.
\end{aligned}
\label{eq:paired-exact-pattern}
\end{equation}

For the hardware test the working point and revealing protocol were selected
entirely before hardware acquisition by an exact search followed by
finite-shot and device-noise rehearsal. The frozen working point was
\begin{equation}
\begin{aligned}
q_*^{(4/5/6)} =
(&2.4162013977,\,
  2.4162013977,\\
 &0.5081592473,\,
  2.4683765918,\\
 &1.5707963268,\,
  2.0656323341).
\end{aligned}
\label{eq:qstar456}
\end{equation}
The restricted protocol was $ZZZ$, and the revealing protocol was the same
two-setting family
\begin{equation}
 \mathcal M_{\rm reveal}=\{ZZZ,YXX\}.
\label{eq:paired-reveal}
\end{equation}

The paired finite-shot analysis used separate null calibrations for the fifth
and sixth revealed directions. For the $4\to5$ test, both restricted circuit
families have exact Fisher rank four, so at $\alpha=0.01$ the conservative
null threshold was defined by
\begin{equation}
\begin{aligned}
\lambda_{{\rm null},5}^{(1-\alpha)}
={}&
\max\!\Bigl\{
Q_{1-\alpha}\!\left[
\lambda_5(F_{ZZZ}^{\rm orig})
\right],\\
&\qquad
Q_{1-\alpha}\!\left[
\lambda_5(F_{ZZZ}^{\rm corr})
\right]
\Bigr\}.
\end{aligned}
\label{eq:paired-null5}
\end{equation}

For the $5\to6$ test, the original revealing family has exact Fisher rank
five and therefore provides the corresponding sixth-direction finite-shot
null channel:
\begin{equation}
 \lambda_{{\rm null},6}^{(1-\alpha)}
 =
 Q_{1-\alpha}
 \!\left[
 \lambda_6(F_{\rm reveal}^{\rm orig})
 \right],
 \qquad
 \alpha=0.01.
 \label{eq:paired-null6}
\end{equation}
The corrected-family sixth eigenvalue
$\lambda_6(F_{\rm reveal}^{\rm corr})$ is the signal statistic for the
$5\to6$ test. These definitions are used throughout the paired cross-backend,
shot-count, and temporal-replication analyses below.

The design criterion explicitly prevented the rank-four restricted spectrum
from becoming numerically marginal. At the exact level the restricted fourth
eigenvalues satisfy
\begin{equation}
 \frac{\lambda_4^{\rm orig,ZZZ}}{\lambda_1^{\rm orig,ZZZ}}
 =0.542,
 \qquad
 \frac{\lambda_4^{\rm corr,ZZZ}}{\lambda_1^{\rm corr,ZZZ}}
 =0.640,
\end{equation}
while the revealed fifth and sixth directions satisfy
\begin{equation}
 \frac{\lambda_5^{\rm orig,reveal}}{\lambda_1^{\rm orig,reveal}}
 =0.441,
 \qquad
 \frac{\lambda_6^{\rm corr,reveal}}{\lambda_1^{\rm corr,reveal}}
 =0.441.
\end{equation}
The exact spectra are therefore strongly separated from numerical rank
thresholds:
\begin{equation}
 \lambda(F_{ZZZ}^{\rm orig})
 \simeq
 (1.000,\,
 0.846,\,
 0.611,\,
 0.542,\,
 0,\,
 0),
\end{equation}
\begin{equation}
 \lambda(F_{\rm reveal}^{\rm orig})
 \simeq
 (0.694,\,
 0.500,\,
 0.306,\,
 0.306,\,
 0.306,\,
 0),
\end{equation}
and
\begin{equation}
 \lambda(F_{ZZZ}^{\rm corr})
 \simeq
 (0.846,\,
 0.779,\,
 0.611,\,
 0.542,\,
 0,\,
 0),
\end{equation}
\begin{equation}
\begin{aligned}
\lambda(F_{\rm reveal}^{\rm corr})
\simeq{}&
(0.694,\,
 0.583,\,
 0.306,\\
&0.306,\,
 0.306,\,
 0.306).
\end{aligned}
\end{equation}

Before submission the complete $52$-circuit workload was rehearsed with an IBM
FakeBackend noise model. At $1024$ shots per circuit the simulated $4\to5$
signal had median $\lambda_5=0.173$ against a $99\%$ null threshold $0.00879$,
while the simulated $5\to6$ signal had median $\lambda_6=0.130$ against a
$99\%$ null threshold $0.00528$. Both detection powers were $1.000$. The
fifth-percentile restricted fourth eigenvalues remained $0.361$ and $0.213$,
respectively, far above the common rank-four null floor.

The paired experiment was then executed without changing the frozen working
point, measurement protocols, parameter-shift rule, shot count, or statistical
decision criterion on three distinct 156-qubit IBM Quantum Heron processors:
\texttt{ibm\_kingston}, \texttt{ibm\_marrakesh}, and \texttt{ibm\_fez}. For
each backend a connected three-qubit chain was selected from the immediately
preceding calibration snapshot and then fixed for both circuit families. The
physical chains were \[ (140,141,142),\qquad (53,54,55),\qquad (21,22,23), \]
respectively.

Each family required
\begin{equation}
 2(1+2\times6)=26
\end{equation}
circuits, so every backend received a single $52$-circuit IBM Runtime Sampler
job with $1024$ shots per circuit, corresponding to $53248$ shot executions.
On every processor all compiled circuits used only the selected physical
three-qubit chain and exactly two native two-qubit gates, with no routing
SWAPs. The reported QPU charge time was $16\,\mathrm{s}$ for each paired job.

The original family reproduced the $4\to5$ pattern on all three processors.
The observed fifth revealing eigenvalues were
\begin{equation}
 \lambda_5^{\rm orig,reveal}
 =
 \begin{cases}
 0.27956, & \text{\texttt{ibm\_kingston}},\\
 0.20496, & \text{\texttt{ibm\_marrakesh}},\\
 0.24712, & \text{\texttt{ibm\_fez}},
 \end{cases}
\end{equation}
while the corresponding conservative $99\%$ fifth-direction null thresholds
defined by Eq.~\eqref{eq:paired-null5} were
\begin{equation}
 \lambda_{{\rm null},5}^{99\%}
 =
 \begin{cases}
 0.02175, & \text{\texttt{ibm\_kingston}},\\
 0.01234, & \text{\texttt{ibm\_marrakesh}},\\
 0.01804, & \text{\texttt{ibm\_fez}}.
 \end{cases}
\end{equation}
The bootstrap detection power was $1.000$ on every backend.

The corrected family simultaneously reproduced the $4\to6$ pattern on all
three processors. The observed sixth revealing eigenvalues were
\begin{equation}
 \lambda_6^{\rm corr,reveal}
 =
 \begin{cases}
 0.23317, & \text{\texttt{ibm\_kingston}},\\
 0.20297, & \text{\texttt{ibm\_marrakesh}},\\
 0.23250, & \text{\texttt{ibm\_fez}},
 \end{cases}
\end{equation}
against the sixth-direction $99\%$ null thresholds defined by
Eq.~\eqref{eq:paired-null6}
\begin{equation}
 \lambda_{{\rm null},6}^{99\%}
 =
 \begin{cases}
 0.008281, & \text{\texttt{ibm\_kingston}},\\
 0.008533, & \text{\texttt{ibm\_marrakesh}},\\
 0.01079, & \text{\texttt{ibm\_fez}}.
 \end{cases}
\end{equation}
Again, the bootstrap detection power was $1.000$ on every backend.

The restricted rank-four structure itself remained strongly resolved
throughout the cross-backend experiment. The observed fourth restricted
eigenvalues were
\begin{equation}
 \lambda_4^{\rm orig,ZZZ}
 =
 (0.56020,\;0.40368,\;0.53480)
\end{equation}
and
\begin{equation}
 \lambda_4^{\rm corr,ZZZ}
 =
 (0.41341,\;0.38581,\;0.47794)
\end{equation}
for Kingston, Marrakesh and Fez, respectively. In every case these values, and
their fifth-percentile bootstrap counterparts, remained far above the
associated fifth-eigenvalue null floor.

The observed hardware spectra illustrate that the effect is not produced by a
marginal eigenvalue. For the corrected family the revealing spectra were
\begin{equation}
\begin{aligned}
\lambda_{\rm Kingston}
  (F_{\rm reveal}^{\rm corr})
&\simeq
(0.719,\,0.543,\,0.325,\\
&\qquad 0.288,\,0.272,\,0.233),
\end{aligned}
\end{equation}

\begin{equation}
\begin{aligned}
\lambda_{\rm Marrakesh}
  (F_{\rm reveal}^{\rm corr})
&\simeq
(0.552,\,0.515,\,0.292,\\
&\qquad 0.262,\,0.233,\,0.203),
\end{aligned}
\end{equation}

\begin{equation}
\begin{aligned}
\lambda_{\rm Fez}
  (F_{\rm reveal}^{\rm corr})
&\simeq
(0.588,\,0.496,\,0.307,\\
&\qquad 0.273,\,0.245,\,0.233).
\end{aligned}
\end{equation}
Thus the sixth direction is comparable in scale to the fourth and fifth
directions rather than a weak spectral tail.

The change in the $q_2$ direction supplies an internal control for the
mechanism. In the original family $q_2$ is an exact global-phase direction,
whereas in the corrected family the replacement
\[
 R_z^{(2)}(q_2)\longrightarrow R_y^{(2)}(q_2)
\]
makes the same parameter physically observable. Across the three processors
the bootstrap median corrected-to-original ratios of the diagonal Fisher
contribution $F_{q_2q_2}$ were approximately
\[
 85.8,\qquad 88.0,\qquad 67.1
\]
for Kingston, Marrakesh and Fez, respectively. 
The appearance of the sixth direction is therefore consistent with, and
quantitatively tracks, the deliberate removal of the original global-phase
redundancy.
The cross-backend paired result is consequently
\begin{equation}
\boxed{
\begin{aligned}
 \text{original family:}\quad &4\to5,\\
 \text{corrected family:}\quad &4\to6,
\end{aligned}}
\label{eq:hardware-4to5-4to6}
\end{equation}
on all three hardware backends, under the same frozen working point and
measurement design.

This paired result strengthens Eq.~\eqref{eq:hardware-negative}. A restricted
Fisher rank of four does not merely fail to exclude one hidden direction: the
same restricted spectrum can arise from families with different larger
locally resolvable ranks. In the present controlled example, the identical
restricted rank is compatible with rank five and with rank six. Therefore
neither the value four nor the stability of the fourth restricted eigenvalue
determines the full local rank without an independent completeness argument.

The logical scope is unchanged. These are controlled statements about
measurement completeness in parametrized quantum families. They do not
identify either family with the primary operational structure, and they do not
establish PR0, PI, closure-faithfulness, or Hypothesis H.

The paired hardware results are summarized in Table~\ref{tab:hardware-paired}.

\begin{table*}[htbp]
\centering
\caption{
    Cross-backend paired IBM Quantum hardware control. The working point,
    restricted and revealing protocols, parameter-shift rule, shot count, and
    statistical decision criterion were fixed before hardware acquisition and
    unchanged across the three processors
}
\label{tab:hardware-paired}
\renewcommand{\arraystretch}{1.15}
\begin{tabular}{lccc}
\toprule
\textbf{Quantity}
& \textbf{Kingston}
& \textbf{Marrakesh}
& \textbf{Fez}\\
\midrule
Physical qubits
& $(140,141,142)$
& $(53,54,55)$
& $(21,22,23)$\\
Original $\lambda_4^{ZZZ}$
& $0.56020$
& $0.40368$
& $0.53480$\\
Corrected $\lambda_4^{ZZZ}$
& $0.41341$
& $0.38581$
& $0.47794$\\
Original revealed $\lambda_5$
& $0.27956$
& $0.20496$
& $0.24712$\\
$99\%$ null for $\lambda_5$
& $0.02175$
& $0.01234$
& $0.01804$\\
Corrected revealed $\lambda_6$
& $0.23317$
& $0.20297$
& $0.23250$\\
$99\%$ null for $\lambda_6$
& $0.008281$
& $0.008533$
& $0.01079$\\
Power, $4\to5$
& $1.000$
& $1.000$
& $1.000$\\
Power, $5\to6$
& $1.000$
& $1.000$
& $1.000$\\
QPU charge time
& $16\,\mathrm{s}$
& $16\,\mathrm{s}$
& $16\,\mathrm{s}$\\
\bottomrule
\end{tabular}
\end{table*}

\subsection{Shot-count robustness of the paired hardware result}

To test whether the paired $4\to5$ and $4\to6$ conclusions depend sensitively
on the chosen sampling budget, the complete frozen paired workload was
repeated on \texttt{ibm\_marrakesh} over an eightfold range of shot counts.
The working point $q_*^{(4/5/6)}$, the physical qubit chain $(53,54,55)$, the
restricted and revealing measurement protocols, the parameter-shift rule, the
transpilation constraints and the bootstrap decision criterion were kept
fixed. Four independent hardware jobs were executed with
\begin{equation}
 N_{\rm shots}\in\{256,512,1024,2048\}
\end{equation}
shots per circuit.

The purpose of this test was not to require the individual observed Fisher
eigenvalues to vary monotonically with the shot count: the four datasets were
acquired in independent hardware jobs and therefore contain independent
statistical fluctuations. Rather, the pre-specified criterion was that the
fifth and sixth revealing eigenvalues remain separated from their
corresponding empirical null spectra, and that the restricted fourth direction
remain resolved, throughout the shot-count range.

For the original family the observed fifth revealing eigenvalues and their
$99\%$ null thresholds were
\begin{equation}
\begin{array}{c|cccc}
N_{\rm shots}
&256&512&1024&2048\\
\hline
\lambda_5^{\rm orig,reveal}
&0.2811&0.1930&0.2103&0.2136\\
\lambda_{{\rm null},5}^{99\%}
&0.07405&0.02684&0.01111&0.00704
\end{array}
\label{eq:shotscaling5}
\end{equation}
and the bootstrap detection power was $1.000$ at every shot count.

For the corrected family the sixth revealing eigenvalues remained similarly
resolved:
\begin{equation}
\begin{array}{c|cccc}
N_{\rm shots}
&256&512&1024&2048\\
\hline
\lambda_6^{\rm corr,reveal}
&0.1786&0.2055&0.2173&0.2256\\
\lambda_{{\rm null},6}^{99\%}
&0.05195&0.01771&0.01690&0.00381
\end{array}
\label{eq:shotscaling6}
\end{equation}
with bootstrap detection power again equal to $1.000$ throughout.

The lowest-statistics case already provides a conservative robustness check.
At $256$ shots per circuit,
\begin{equation}
 Q_{0.05}\!\left(\lambda_5^{\rm orig,reveal}\right)=0.2045
 >0.07405
 =\lambda_{{\rm null},5}^{99\%},
\end{equation}
while
\begin{equation}
 Q_{0.05}\!\left(\lambda_6^{\rm corr,reveal}\right)=0.1313
 >0.05195
 =\lambda_{{\rm null},6}^{99\%}.
\end{equation}
The restricted rank-four guard also remained satisfied at this lowest shot
count:
\begin{equation}
\begin{aligned}
 Q_{0.05}\!\left(\lambda_4^{\rm orig,ZZZ}\right)
 &=0.3225,\\
 Q_{0.05}\!\left(\lambda_4^{\rm corr,ZZZ}\right)
 &=0.2211.
\end{aligned}
\end{equation}
both well above the common fifth-eigenvalue null threshold.

As the shot count increased, the empirical null spectra generally contracted
while the additional revealed directions remained of order $10^{-1}$. The
median signal-to-null ratios for the $4\to5$ test were approximately
\begin{equation}
 10.6,\quad21.6,\quad57.2,\quad86.5,
\end{equation}
for $256$, $512$, $1024$ and $2048$ shots, respectively. For the $5\to6$ test
they were
\begin{equation}
 8.05,\quad25.7,\quad24.7,\quad138.4.
\end{equation}
The small non-monotonicity between individual runs is expected because the
four shot-count datasets are independent hardware acquisitions rather than
nested resamplings of a single raw dataset.

Thus the paired $4\to5$ and $4\to6$ distinction is not tied to the reference
budget of $1024$ shots per circuit. Both additional directions remain resolved
throughout the tested range from $256$ to $2048$ shots, with
bootstrap-estimated detection power $1.000$ at every sampling budget. The
negative-control conclusion is therefore robust against the tested level of
finite-shot variation.

\subsection{Temporal replication across calibration epochs}

To test whether the paired $4\to5$ and $4\to6$ result depends on a particular
hardware calibration snapshot, the same frozen paired workload was repeated on
\texttt{ibm\_kingston} approximately two days after the original Kingston
acquisition. No working-point, protocol, shot-count or statistical criterion
was re-optimized.

The original Kingston paired run used the physical chain \[ (140,141,142), \]
whereas the later calibration snapshot selected a different optimal connected
chain, \[ (122,123,124). \] The repeated job again contained the same $52$
circuits with $1024$ shots per circuit, corresponding to $53248$ shot
executions. As before, all compiled circuits were confined to the selected
three-qubit chain and used exactly two native two-qubit gates, with no routing
SWAPs. The reported QPU charge time was $16\,\mathrm{s}$.

The original circuit family again exhibited a robust fifth revealing
direction,
\begin{equation}
 \lambda_5^{\rm orig,reveal}=0.24028,
\end{equation}
while the conservative $99\%$ fifth-direction null threshold was
\begin{equation}
 \lambda_{{\rm null},5}^{99\%}=0.01345.
\end{equation}
The corresponding bootstrap detection power was $1.000$.

The corrected family simultaneously reproduced the sixth revealing direction,
\begin{equation}
 \lambda_6^{\rm corr,reveal}=0.19021,
\end{equation}
against a sixth-direction $99\%$ null threshold
\begin{equation}
 \lambda_{{\rm null},6}^{99\%}=0.00762,
\end{equation}
again with bootstrap-estimated detection power $1.000$.

The restricted rank-four structure remained well separated from the empirical
null spectrum:
\begin{equation}
 \lambda_4^{\rm orig,ZZZ}=0.48124,
 \qquad
 \lambda_4^{\rm corr,ZZZ}=0.34895.
\end{equation}
The corresponding fifth-percentile bootstrap values,
\begin{equation}
\begin{aligned}
 Q_{0.05}\!\left(\lambda_4^{\rm orig,ZZZ}\right)
 &=0.3893,\\
 Q_{0.05}\!\left(\lambda_4^{\rm corr,ZZZ}\right)
 &=0.2813.
\end{aligned}
\end{equation}
remained far above the common fifth-eigenvalue null floor.

The observed spectra were
\begin{equation}
\begin{aligned}
&\lambda(F_{ZZZ}^{\rm orig})\\
&\qquad\simeq
(0.944,\,
 0.806,\,
 0.506,\\
&\qquad\phantom{\simeq(}
 0.481,\,
 0.00163,\,
 0.00052).
\end{aligned}
\end{equation}

\begin{equation}
\begin{aligned}
&\lambda(F_{\rm reveal}^{\rm orig})\\
&\qquad\simeq
(0.606,\,
 0.476,\,
 0.326,\\
&\qquad\phantom{\simeq(}
 0.255,\,
 0.240,\,
 0.00120).
\end{aligned}
\end{equation}
and
\begin{equation}
\begin{aligned}
&\lambda(F_{ZZZ}^{\rm corr})\\
&\qquad\simeq
(0.755,\,
 0.700,\,
 0.467,\\
&\qquad\phantom{\simeq(}
 0.349,\,
 0.00301,\,
 0.00044).
\end{aligned}
\end{equation}
\begin{equation}
\begin{aligned}
&\lambda(F_{\rm reveal}^{\rm corr})\\
&\qquad\simeq
(0.591,\,
 0.531,\,
 0.301,\\
&\qquad\phantom{\simeq(}
 0.254,\,
 0.232,\,
 0.190).
\end{aligned}
\end{equation}

The internal $q_2$ control was also reproduced. In the original family the
revealing diagonal Fisher contribution remained at the empirical null level,
\begin{equation}
 F_{q_2q_2}^{\rm orig}\simeq0.00292,
\end{equation}
whereas in the corrected family it was
\begin{equation}
 F_{q_2q_2}^{\rm corr}\simeq0.497.
\end{equation}
The bootstrap median corrected-to-original ratio was approximately $82.8$.

Thus the frozen paired distinction
\begin{equation}
 4\to5
 \qquad\hbox{versus}\qquad
 4\to6
\end{equation}
was reproduced on the same processor in a later calibration epoch and on a
different physical three-qubit chain. Together with the three-backend
replication and the shot-count study, this
supports robustness of the negative-control result across the tested processor
identities, sampling budgets, and calibration epochs.

Figure~\ref{fig:ibm-robustness} summarizes the three robustness checks and the
internal $q_2$ control. Across processor identity, sampling budget, and
calibration epoch, the additional fifth and sixth Fisher directions remain
separated from their empirical null spectra, while the $q_2$ control tracks
the deliberate removal of the global-phase redundancy.

\begin{figure*}[t!]
    \centering
    \includegraphics[width=\textwidth]{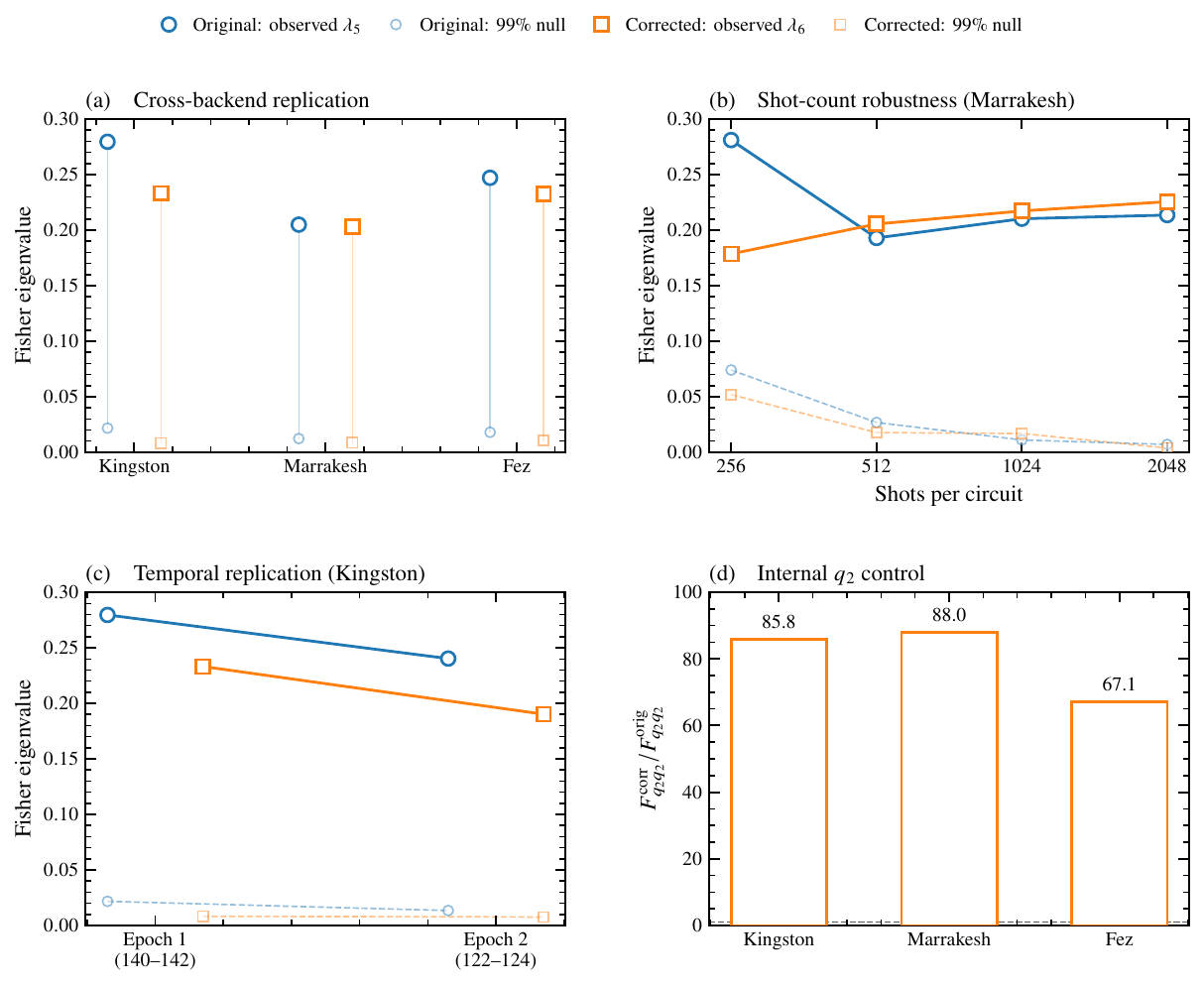}
    \caption{
        Hardware robustness of the paired Fisher-rank negative control.
        (a) Cross-backend replication at 1024 shots per circuit.
        (b) Shot-count robustness on \texttt{ibm\_marrakesh}.
        (c) Temporal replication on \texttt{ibm\_kingston} under two calibration
        snapshots and different physical-qubit chains.
        (d) Bootstrap-median internal $q_2$ control.
        Across all tests, the paired $4\to5$ and $4\to6$ distinctions remain resolved
    }
    \label{fig:ibm-robustness}
\end{figure*}

\section{Requirements for an experimental test of primary operational rank}
\label{sec:experiment}

The hardware negative controls of Section~\ref{sec:calibration} address only
the measurement-completeness part of the rank-preserving chain. They show that
the same restricted Fisher rank four can conceal larger locally resolvable
ranks, including different larger ranks in closely related parametrized
families, and that this conclusion is robust across the tested processors,
shot counts, and calibration epochs. These experiments therefore rule out the
inference of full local rank from a restricted Fisher spectrum alone.

They do not, however, constitute a measurement of primary operational rank.
The stronger question is what additional structure must be fixed and verified
before an observed Fisher rank can be interpreted as a test of the primary
operational sector.

A test of primary operational rank must reverse the logic of a naive
rank measurement. The primary operational structure must first be specified
independently of the measured spectrum and satisfy PR0; when an equality rather
than a lower bound is claimed from the compositional construction, Hypothesis
H must also be established. A physically motivated observer map $\Piobs$ and
the relevant scale must then be fixed independently. The resulting smooth
realization must satisfy B2-R and RB1--RB4, and PI must identify the observer
image with the primary integral distribution. Only after these steps may a
measurement protocol be used to test the rank. For a positive determination,
the admissible protocol family must additionally resolve the primary leaf and
satisfy the support and finite-mixture assumptions required by the Fisher
bridge.

The two directions of the test are not symmetric, and separating them is what
makes the programme experimentally usable. Refuting an operational-rank claim
requires strictly weaker measurement conditions than establishing one.

\paragraph{Negative test.}
A claimed primary operational rank $r$ can be falsified without establishing
global measurement completeness. One must first identify, independently of
the measured Fisher spectrum, at least $r+1$ candidate primary directions
within the same operational realization and verify the part of the bridge
needed to associate their smooth images with that realization. The primary
sector, observer map $\Piobs$, physical scale, and relevant
closure-faithfulness relations must therefore be fixed in advance.

If an admissible measurement protocol resolves the corresponding $r+1$
smooth directions as Fisher-independent, then the primary smooth rank is at
least $r+1$ and the claimed operational rank $r$ is incompatible with that
closure-faithful realization. Global informational completeness is not
required, because measurement incompleteness can remove resolved directions
but cannot create additional independent ones.

Corollary~\ref{cor:onesided} provides a complementary one-sided
diagnostic. For the exact Fisher metric,
\begin{equation}
    \deff(F)>r
    \quad\Longrightarrow\quad
    \rank F\ge r+1.
\end{equation}
Thus a participation ratio above the claimed rank can also falsify
that rank, provided the measured directions have been independently
tied to the same primary operational realization. For an experimental
estimate, the excess above $r$ must remain robust under the corresponding
finite-sample uncertainty analysis or null model. No assumption of
global informational completeness is required for this one-sided
inference.

This is stronger than merely observing $\rank F>r$ somewhere in a
larger parametrized or observer space. As
Remark~\ref{rem:threeprimary} shows, $\Xprimary$ may contain tangent
directions outside the primary integral leaf unless PI, or an equivalent 
independently justified membership condition, is supplied. A negative test 
therefore requires the measured directions to be tied to the same primary
operational realization before their Fisher independence can falsify
its claimed rank.

\paragraph{Positive test.}
Establishing that the primary operational rank equals $r$ is more
demanding. The complete bridge is required: the operational closure
must possess a real, closure-faithful smooth realization of rank $r$;
RB1--RB4 must hold; and PI must identify the observer image locally
with the primary leaf. The admissible protocol family must then locally
resolve that image. Finite labeled-mixture closure and locally constant
outcome support are additionally required to represent the available
resolving power by a finite Fisher matrix of rank $r$.

Without these completeness conditions, an observed Fisher rank $r$
remains compatible with unresolved additional directions, exactly as
demonstrated by the negative controls of
Section~\ref{sec:calibration}.

The asymmetry is therefore fundamental: rank excess can provide a
one-sided falsification once operational membership is established,
whereas rank equality requires completeness of the entire
operational-to-smooth-to-Fisher chain.

Existing bounds on extra-dimensional physics do not by themselves
determine $\rank F_{\rm primary}$; translating such observations into
an operational-rank statement requires precisely the sector
identification and bridge conditions specified above.

For quantum-gate implementations, derivative estimation should use exact
parameter-shift relations where available rather than small finite
differences, which amplify shot noise as $1/h^2$ \cite{schuld,mari2021}. 
Small Fisher eigenvalues should be assessed with bootstrap or an explicit 
null model, not a universal fixed cutoff. Zero-probability outcomes require
support-aware treatment, and the Akaike information criterion (AIC) or the 
Bayesian information criterion (BIC) should not be applied directly to the 
number of nonzero Fisher eigenvalues without explicit competing likelihood 
models \cite{burnham2004}.

\section{Discussion}

\subsection{What is established}

The central result of this work is a conditional bridge between three
a priori distinct notions of dimension: compositional operational rank,
smooth tangent rank, and measurable Fisher rank. The bridge is formulated for
an arbitrary finite rank $r$ and does not rely on a four-dimensional input.
Its non-trivial structural step is the smooth realization: preservation of the
total rank is insufficient, and the operational closure must instead be
represented closure-faithfully. Under real representability, constant rank,
involutivity, and observer/scale compatibility, the resulting primary
distribution integrates locally to a leaf whose dimension equals the
operational rank.

The passage from this smooth rank to Fisher information requires an additional
physical identification. The observer map must be fixed independently of the
measured spectrum, its tangent image must coincide with the primary
distribution through PI, and the admissible protocol family must locally
resolve that sector. Closure under finite labeled mixtures then allows the
available resolving power to be represented by a finite composite protocol,
while locally constant support ensures that the differential and Fisher ranks
coincide.

For the rank-four specialization, the complete implication is therefore
conditional:
\begin{equation}
\begin{aligned}
&\mathrm{PR0}+H+\text{rank-preserving bridge}+\mathrm{PI}\\
&\quad+\text{local resolution}
+\text{finite labeled mixtures}\\
&\quad+\text{locally constant support}
\Longrightarrow
\rank F_{\rm primary}=4.
\end{aligned}
\end{equation}
The value four is supplied by the rank-four construction of
Ref.~\cite{before}; it is not selected by the bridge itself. The resulting
rank belongs to the specified primary observer sector and should not be
identified with the dimension of every microscopic, quantum, or mathematical
representation of the same system.

The dimensional-reduction application adds a further structural conclusion.
Producing four effective directions is not sufficient: the reduction must
also induce or independently specify the operational structure whose rank is
being represented and preserve the relevant dependency relations. The
examples of Section~\ref{sec:higher} show three distinct possibilities.
Deletion or contraction in a freely generated commuting structure is benign;
compact-sector observables that survive the reduction can contribute
additional operational directions; and non-Abelian compositional generation
can obstruct closure-faithfulness altogether.

The hardware experiments establish a complementary, deliberately narrower
result at the measurement end of the chain. A restricted Fisher rank can
remain numerically stable while failing to reveal additional locally
distinguishable directions. The observed $4\to5$ and $4\to6$ paired controls
therefore support the need to treat measurement completeness as an
independent assumption rather than infer it from the measured rank itself.

\subsection{Scope and limitations}
\label{sec:not-established}

The results above are intentionally conditional, and several boundaries of
their interpretation are important.

Hypothesis H remains an additional physical hypothesis. The
rank-preserving smooth bridge neither proves H nor derives the rank-four input
from first principles. Likewise, the observer map $\Piobs$ is not derived
uniquely. Its role is to represent the physically specified passage from a
microscopic description to the primary observer sector, but a concrete theory
must supply that map independently of the Fisher spectrum.

No universal upper bound on microscopic Fisher or quantum Fisher rank follows.
A microscopic quantum family may possess more than four locally
distinguishable directions without contradiction. Nor is operational rank
four, by itself, a theorem about the topological dimension or Lorentzian
signature of spacetime. The present construction fixes a number of primary
operational directions. The separation into three spatial and one temporal
direction requires additional structure. In Ref.~\cite{before}, this role is
assigned to the separate condition C5 of dynamic causality, which selects the
Lorentzian real form $SO(3,1)$ from the pre-causal $SO(4)$ structure; C5 is
not used in the bridge developed here.

Existing experimental bounds on extra dimensions do not by themselves
constitute evidence for the bridge. Short-range gravity experiments constrain
specific scenarios \cite{murata}, while collider searches provide
complementary model-dependent limits \cite{pdgextra2023}. Proposed quantum
probes offer additional routes by which extra dimensions could leave
observable signatures \cite{feng2024extra,mazharimousavi2025}.

Such measurements make a rank-four primary sector physically relevant.
However, standard models with extra dimensions already recover the observed
macroscopic regime. The observed count therefore does not distinguish the
present criterion from conventional reduction mechanisms.

The dimensional-reduction criterion has also not yet excluded a concrete
existing construction. This limitation is structural rather than merely the
absence of an example. Proposition~\ref{prop:free} shows that for freely
generated commuting accessibility structures, compactness itself does not
lower operational rank, while deletion and contraction preserve the free
structure. Consequently, a broad class of simple reductions is necessarily
benign with respect to RB3. The circle reduction of
Section~\ref{sec:kkcircle} belongs to this class.

At the opposite extreme, some non-Abelian realizations violate
closure-faithfulness directly: finite compositions can make an elementary
transition operationally dependent while its tangent representative remains
outside the linear span of the generating directions. Such systems fall
outside the rank-preserving bridge rather than furnishing a failure of the
criterion from within its domain. A fully discriminating example would
instead satisfy the bridge's admissibility assumptions and the required
reduced count while failing the dependency-preservation test. No such concrete
higher-dimensional reduction has been exhibited here.

A separate limitation concerns locality. The bridge describes
\emph{differential} operational distinguishability. Sections~\ref{sec:smooth}
and \ref{sec:observer} represent accessible directions through differentials,
so genuinely global distinctions can escape the construction. A discrete
holonomy sector, topological label, or superselection charge may distinguish
states even though its differential vanishes on every connected local leaf.
Such information contributes neither to $\dop$ nor to a Fisher matrix,
however complete the local protocol family may be. The discrete Wilson-line
case of Section~\ref{sec:wilsonlines} provides an explicit example.

The operational structure may therefore contain two qualitatively different
types of information: continuous distinctions represented by local tangent
directions, and global labels distinguishing disconnected components or
sectors. The present rank-preserving bridge addresses only the first. A
framework incorporating both would require a global extension of the local
rank used here.

Finally, the criterion does not derive a unique compactification geometry,
Calabi--Yau manifold, brane construction, or other specific reduction
mechanism. Its role is instead to supply an additional operational condition
that candidate mechanisms may satisfy or fail once their observer sector,
admissible procedures, and physical scale have been specified.

\subsection{Relation to scale-dependent distinguishability geometry}
\label{sec:scale}

The companion distinguishability-geometry framework \cite{dg} provides a
natural language for the scale dependence already encountered in
Section~\ref{sec:higher}. It defines a smooth local distinguishability metric
through
\begin{equation}
 D(x,x+dx)=G_{AB}(x)\,dq^A dq^B+O(\|dq\|^3),
 \label{eq:Dmetric}
\end{equation}
and allows the spectrum of distinguishable directions to vary with
observational scale. A microscopic description can therefore contain more
locally distinguishable directions than a coarse primary observer sector
without requiring the microscopic state space itself to become
four-dimensional.

Two elementary results clarify the relation between this scale-dependent
geometry and the rank used in the present work. The first compares the
participation-ratio effective dimension with matrix rank; the second shows that
admissible coarse-graining cannot increase operational rank.

\begin{proposition}[Effective dimension and rank]
\label{prop:deffrank}

Let $G(\ell)$ be a non-zero local distinguishability metric at observational
scale $\ell$, with eigenvalues $\lambda_i(\ell)\ge0$, and define
\begin{equation}
 \deff(\ell)
 =
 \frac{(\sum_i\lambda_i)^2}{\sum_i\lambda_i^2}
 =
 \frac{(\operatorname{Tr}G(\ell))^2}
      {\operatorname{Tr}(G(\ell)^2)}.
 \label{eq:deffG}
\end{equation}
Then
\begin{equation}
 \deff(\ell)\le\rank G(\ell),
 \label{eq:deffbound}
\end{equation}
with equality if and only if all non-zero eigenvalues coincide.

\end{proposition}

\begin{proof}

Let
\[
S=\{i:\lambda_i>0\},
\qquad
r=|S|=\rank G(\ell).
\]
Applying the Cauchy--Schwarz inequality to
$(\lambda_i)_{i\in S}$ and the constant vector of length $r$ gives
\begin{equation}
 \left(\sum_{i\in S}\lambda_i\right)^2
 \le
 r\sum_{i\in S}\lambda_i^2.
\end{equation}
Eigenvalues outside $S$ vanish, so the sums may be extended to all $i$.
Division by $\sum_i\lambda_i^2>0$ yields
Eq.~\eqref{eq:deffbound}. Equality in Cauchy--Schwarz holds exactly when all
non-zero $\lambda_i$ are equal.

\end{proof}

For a parametrized statistical family, $G$ may be taken to be the classical
Fisher matrix of the chosen protocol. For a pure-state quantum family equipped
with the Fubini--Study metric, the corresponding metric is one quarter of the
quantum Fisher information matrix in the convention used here.

Unlike rank, however, the participation ratio is not invariant under an
arbitrary reparametrization of the state space. Under $q\to q'$,
\[
G\longrightarrow J^{\mathsf T}GJ,
\]
so the spectrum and hence $\deff$ generally change while the rank does not.
The coordinate choice must therefore be included in the specification of the
sector in which $\deff$ is interpreted. Table~\ref{tab:fisher}, for example,
uses the natural gate-angle coordinates of circuit~\eqref{eq:circuit}.

Rank and $\deff$ consequently answer different questions. Rank counts the
number of directions that are locally distinguishable at all, whereas
$\deff$ is a weighted spectral diagnostic that decreases as those directions
become anisotropic. For a rank-four sector with four strictly positive
eigenvalues,
\begin{equation}
 1<\deff\le4.
\end{equation}
Values arbitrarily close to one are compatible with rank four under strong
anisotropy, whereas $\deff=4$ requires equality of the four non-zero
eigenvalues. Thus an observed $\deff<4$ is not evidence against a
four-dimensional local rank.

\begin{corollary}[A one-sided rank bound that needs no completeness]
\label{cor:onesided}

Let $G=F^{(M)}$ be the Fisher metric of any admissible protocol, or of any
finite labeled composite of admissible protocols, on a region where
Lemma~\ref{lem:fisher} applies. Then
\begin{equation}
 \deff(G)
 \le
 \rank G
 \le
 \dop(q).
 \label{eq:onesided}
\end{equation}
Consequently, $\deff(G)>4$ for the exact pulled-back Fisher metric implies
$\dop(q)\ge5$ for the independently specified sector being probed. For an
experimental estimate of $\deff$, the excess above four must remain robust
under the corresponding finite-sample uncertainty analysis.

\end{corollary}

This inequality has the same one-sided character as the experimental logic of
Section~\ref{sec:calibration}. Measurement incompleteness can reduce the
pulled-back rank but cannot create additional independent directions.
Therefore $\deff>4$, once statistically resolved for an independently
specified sector, supplies a lower bound requiring at least five local
directions. The converse does not hold:
$\deff\le4$ is compatible with a larger rank, and even
$\deff=4$ does not imply rank four. For example,
\[
(1,1,1,1/2,\varepsilon),
\qquad
\varepsilon\approx0.1126,
\]
has participation ratio four while having rank five. The participation ratio
is therefore useful as a lower-bound diagnostic of local smooth rank, not as a
standalone estimator of operational rank.

\begin{proposition}[Rank monotonicity under coarse-graining]
\label{prop:coarse}

Let $\mathcal M_{\rm coarse}$ and $\mathcal M_{\rm fine}$ be admissible
families such that for every $M_c\in\mathcal M_{\rm coarse}$ there exist
$M_f\in\mathcal M_{\rm fine}$ and a parameter-independent stochastic map
$\Lambda_c$ satisfying
\[
\Phi_{M_c}=\Lambda_c\circ\Phi_{M_f}.
\]
Then
\begin{equation}
\begin{aligned}
 \Vnull^{\rm coarse}(q)
 &\supseteq
 \Vnull^{\rm fine}(q),\\
 \dop^{\rm coarse}(q)
 &\le
 \dop^{\rm fine}(q).
\end{aligned}
\label{eq:coarse}
\end{equation}

\end{proposition}

\begin{proof}

Fix $v\in\Vnull^{\rm fine}(q)$ and
$M_c\in\mathcal M_{\rm coarse}$, and choose $M_f$ and $\Lambda_c$ as in the
hypothesis. Since $\Lambda_c$ is independent of the parameters, the chain rule
gives
\[
d\Phi_{M_c,q}
=
\Lambda_c\,d\Phi_{M_f,q}.
\]
Because $v$ is annihilated by $d\Phi_{M_f,q}$, it is also annihilated by
$d\Phi_{M_c,q}$. As $M_c$ was arbitrary,
$v\in\Vnull^{\rm coarse}(q)$, proving the inclusion of null spaces.
Equation~\eqref{eq:dop} then gives the inequality of operational dimensions.
The corresponding Fisher-rank statement follows from
Lemma~\ref{lem:fisher}.

\end{proof}

Proposition~\ref{prop:coarse} formalizes the direction of the scale flow used
throughout the dimensional-reduction discussion. A parameter-independent
coarse-graining can remove locally distinguishable directions but cannot
create new ones. Hence a transition from a higher microscopic rank to a
lower-rank macroscopic observer sector is structurally natural. What the
monotonicity does not determine is the endpoint: nothing in
Proposition~\ref{prop:coarse} selects four rather than three, five, or another
rank. The value four must enter through the independently specified primary
structure or through empirical input.

A final notational distinction is important. The symbol $d_{\rm eff}$ has
different meanings in the two companion constructions. In
Ref.~\cite{before}, it denotes the emergent effective dimension associated
with the minimum self-consistent operational construction, whose physical
value is four under Hypothesis H. In Ref.~\cite{dg}, it denotes the generally
non-integer participation ratio of the eigenvalues of a local
distinguishability metric at observational scale. The present paper reserves
$\dop$ for the protocol-independent local operational dimension of
Section~\ref{sec:smooth} and uses $\deff$ only for the participation-ratio
diagnostic. Equality between these quantities requires additional spectral
and protocol assumptions and should never be inferred from notation alone.

\section{Conclusion}

We have established sufficient conditions under which a compositional
operational rank can be transferred consistently across two changes of
description: first to a smooth tangent realization and then to measurable
Fisher information. The central structural requirement is
closure-faithfulness, which demands preservation of the operational
dependency relations rather than agreement of total ranks alone. Under the
corresponding smooth-realization and observer-identification assumptions,
operational rank, tangent rank, and Fisher rank can therefore be related
without being identified by definition. The construction holds for arbitrary
finite rank $r$.

At $r=4$, this bridge provides a conditional operational criterion for
dimensional reduction in higher-dimensional theories. The value four is 
not selected by the bridge itself: it enters through the rank-four 
construction of Ref.~\cite{before}, where C1--C4 select a representation 
dimension $n=4$. Applying that construction to the primary sector requires 
PR0, while identifying $n$ with the primary operational rank additionally 
requires Hypothesis H. Critical microscopic dimension and observer-accessible 
operational rank therefore remain distinct invariants. In particular,
$\rank F_Q^{\rm micro}>4$ is fully compatible with the framework; the result
places no universal four-dimensional bound on microscopic state spaces.

The explicit higher-dimensional examples separate three qualitatively
different regimes. In freely generated commuting reductions, removal of
compact directions by deletion or contraction can preserve the operational
dependency structure. If compact-sector information survives, as for an
independently accessible Wilson-line holonomy, it can contribute an additional
primary operational direction rather than disappear with the compact
coordinate. In non-Abelian compositional realizations, the bridge can fail
more fundamentally because finite operational generation need not agree with
linear dependence of tangent representatives. Thus obtaining the desired
reduced dimension is weaker than preserving the operational structure that
gives that dimension physical meaning.

The quantum-circuit experiments provide an independent negative control at
the measurement end of the construction. Two closely related six-parameter
families exhibit the same restricted Fisher rank four, while the same
pre-specified enlarged protocol reveals rank five in one family and rank six
in the other. This distinction is reproduced across three IBM Quantum
processors and remains robust under changes of sampling budget and calibration
epoch. The experiments therefore demonstrate that a stable restricted Fisher
rank does not determine the full locally resolvable rank. They support the
need for measurement completeness as an independent assumption, while neither
measuring primary operational rank nor providing evidence for fundamental
four-dimensionality.

The next step is to apply the criterion to a concrete higher-dimensional
model in which the primary sector, elementary transitions, compositional
closure, and observer map are specified independently of the measured
spectrum. The most discriminating case would be a reduction that achieves the
required dimensional count while satisfying the admissibility conditions of
the bridge but nevertheless fails dependency preservation. Such an example
would turn the present structural criterion into a direct test of a specific
higher-dimensional reduction mechanism.

\section*{Statements and Declarations}

\paragraph{Funding.}
The authors received no specific funding for this work.

\paragraph{Competing interests.}
The authors declare no competing interests.

\paragraph{Generative AI and AI-assisted technologies.}
The authors used generative AI tools, including Claude (Anthropic), ChatGPT
(OpenAI), and Gemini (Google), to assist with language editing, stylistic
refinement, consistency of notation and terminology, reference cross-checking,
and the development and debugging of Python code used in the numerical
analysis. All AI-assisted code, calculations, references, and scientific
claims were independently checked and, where necessary, revised by the
authors. The authors take full responsibility for the content of the
manuscript.

\section*{Data Availability Statement}

The raw IBM Quantum counts, calibration snapshots, transpilation metadata,
numerical benchmarks, reconstructed Fisher matrices, bootstrap outputs, and
machine-readable analysis results used in this work are publicly available
in the Zenodo reproducibility package associated with this article
\url{https://doi.org/10.5281/zenodo.22144236}~\cite{repro}.

\section*{Code Availability Statement}

The source code used for the exact Fisher analysis, hardware workloads,
shot-count study, temporal replication, and reconstruction of the reported
results is available in the same Zenodo reproducibility package. Detailed
runtime identifiers, physical-qubit selections, shot budgets, and related
reproducibility metadata are listed in Appendix~\ref{app:reproducibility}.

\appendix

\section{Logical status of the main implications}

For clarity, the complete logical chain used in the paper is
\begin{equation}
\mathrm{C1\text{--}C4}
\Longrightarrow n=4.
\label{eq:app1}
\end{equation}

\begin{equation}
\mathrm{PR0}+H
\Longrightarrow
\rank\Apr([s])=4.
\label{eq:app2}
\end{equation}

\begin{equation}
\begin{aligned}
&\text{Eq.~\eqref{eq:app2}}
+\text{physical realization}
+\mathrm{B2\text{-}R}\\
&\quad+\mathrm{RB3}\\
&\qquad\Longrightarrow
\rank\Dprimary=4.
\end{aligned}
\label{eq:app3}
\end{equation}

\begin{equation}
\begin{aligned}
&\text{Eq.~\eqref{eq:app3}}
+\mathrm{RB1}
+\mathrm{RB2}
+\mathrm{RB4}\\
&\qquad\Longrightarrow
\dim\Lprimary=4.
\end{aligned}
\label{eq:app4}
\end{equation}

\begin{equation}
\begin{aligned}
&\text{Eq.~\eqref{eq:app4}}
+\mathrm{PI}
+\text{local resolution}\\
&\quad+\text{finite-mixture closure}\\
&\quad+\text{locally constant support}\\
&\qquad\Longrightarrow
\rank F_{\rm primary}^{(M_{\rm comp})}=4.
\end{aligned}
\label{eq:app5}
\end{equation}
No implication above yields $\rank F_Q^{\rm arbitrary}\le4$. Each implication
also holds with $4$ replaced throughout by an arbitrary primary rank $r$,
except for the specific input from Ref.~\cite{before} in Eq.~\eqref{eq:app1}.

\section{Compact proof of the composite-protocol lemma}

For a labeled mixture $M_*=\sum_i\lambda_iM_i$ with $\lambda_i>0$ and
$\sum_i\lambda_i=1$, the outcome probability is $p^{(*)}_{ij}=\lambda_i
p^{(i)}_j$. Therefore Fisher information is additive:
\begin{equation}
 F^{(*)}=\sum_i\lambda_iF^{(i)}.
 \label{eq:additiveF}
\end{equation}
For positive-semidefinite matrices,
\begin{equation}
 \Ker\!\left(\sum_i\lambda_iF^{(i)}\right)
 =\bigcap_i\Ker F^{(i)}.
 \label{eq:kernelintersection}
\end{equation}
Starting with $K_0=T_qQ$, whenever $K_m$ strictly contains $\Vnull$ one can
choose an admissible protocol whose kernel does not contain all of $K_m$.
Intersecting with that kernel lowers the dimension by at least one. The
procedure terminates after at most $N-\dim\Vnull$ steps.

\section{A no-go example for generator counting}

Let $\mathfrak h_3=\operatorname{span}\{X,Y,Z\}$ with
\begin{equation}
 [X,Y]=Z,
 \qquad
 [X,Z]=[Y,Z]=0.
 \label{eq:heisenberg}
\end{equation}
The algebra is generated by two elements $X,Y$, so $\mu(\mathfrak h_3)=2$. If
the corresponding Heisenberg group acts on itself by left translations, the
stabilizer is trivial and the orbit dimension is three. Hence
\begin{equation}
 \mu(\mathfrak h_3)=2\neq3=\dim(\mathfrak h_3/\mathfrak h_x).
 \label{eq:heiscount}
\end{equation}
The number of primitive generators therefore does not determine orbit
dimension.

\section{Reproducibility details}
\label{app:reproducibility}

The exact Fisher spectra of Section~\ref{sec:calibration} are reproduced by
the accompanying analysis code, which constructs the circuit of
Eq.~\eqref{eq:circuit}, evaluates analytic state derivatives, and computes the
restricted, tomographically complete product-Pauli, and quantum Fisher
matrices. The 27-setting Pauli protocol is equally weighted, the numerical
support and rank tolerances are those specified in
Eqs.~\eqref{eq:FICweight}--\eqref{eq:ranktol}, and the random-point benchmark
uses seed \texttt{20260815}.

The hardware protocols were fixed before acquisition by local scripts that
identified the rank-revealing Pauli extension, selected the working points,
tested the finite-shot null calibration, and performed device-noise rehearsals.

\paragraph{Initial cross-backend negative control.}
The pre-specified 26-circuit workload of Section~\ref{sec:calibration} was
executed independently on three 156-qubit IBM Quantum processors with
$1024$ shots per circuit:
\begin{center}
\small
\begin{tabular}{@{}ll@{}}
\textbf{Backend / qubits} & \textbf{IBM Runtime job ID}\\
\texttt{ibm\_marrakesh} $(13,14,15)$ &
\texttt{da46dh61vhnc73fkk76g}\\
\texttt{ibm\_fez} $(97,107,108)$ &
\texttt{da46sqbotlns739afuh0}\\
\texttt{ibm\_kingston} $(140,141,142)$ &
\texttt{da46uu61vhnc73fkkra0}
\end{tabular}
\end{center}

\paragraph{Paired $4\to5$ and $4\to6$ control.}
The frozen 52-circuit workload at the working point of
Eq.~\eqref{eq:qstar456}, using the restricted protocol $ZZZ$ and revealing
extension $YXX$, was executed with $1024$ shots per circuit:
\begin{center}
\small
\begin{tabular}{@{}ll@{}}
\textbf{Backend / qubits} & \textbf{IBM Runtime job ID}\\
\texttt{ibm\_kingston} $(140,141,142)$ &
\texttt{da5h84s3jnrc73aha0hg}\\
\texttt{ibm\_marrakesh} $(53,54,55)$ &
\texttt{da5hjle1vhnc73fm3hv0}\\
\texttt{ibm\_fez} $(21,22,23)$ &
\texttt{da5ifmeaa69c739kti60}
\end{tabular}
\end{center}
Each job contained two circuit families, two measurement settings, and
thirteen parameter configurations per setting, for $52$ circuits and
$53248$ shot executions. The IBM-reported QPU charge time was
$16\,\mathrm{s}$ for each backend.

\paragraph{Shot-count robustness.}
The same paired workload was repeated on \texttt{ibm\_marrakesh} using
physical qubits $(53,54,55)$:
\begin{center}
\small
\begin{tabular}{@{}rll@{}}
\textbf{Shots} & \textbf{IBM Runtime job ID} & \textbf{QPU time}\\
256  & \texttt{da5j8pbotlns739c09hg} & $6\,\mathrm{s}$\\
512  & \texttt{da5j92maa69c739kuebg} & $9\,\mathrm{s}$\\
1024 & \texttt{da5j9es3jnrc73ahc7og} & $16\,\mathrm{s}$\\
2048 & \texttt{da5j9s6aa69c739kuf5g} & $30\,\mathrm{s}$
\end{tabular}
\end{center}

\paragraph{Temporal replication.}
The frozen paired workload was repeated on \texttt{ibm\_kingston} under a
later calibration snapshot. The selected physical chain changed from
$(140,141,142)$ to $(122,123,124)$. The run again used $52$ circuits with
$1024$ shots per circuit. The IBM Runtime job identifier was
\texttt{da6mrrc6l22c73dmhoag}, and the reported QPU charge time was
$16\,\mathrm{s}$.

For every hardware dataset, the reproducibility archive contains the frozen
working point and protocol definition, logical and transpiled circuits,
calibration snapshots, transpilation metadata, raw counts, IBM Runtime job
metadata, reconstructed Fisher matrices, bootstrap null and signal
distributions, and machine-readable analysis reports. The complete archived
package is available at
\url{https://doi.org/10.5281/zenodo.22144236}~\cite{repro}.

\bibliographystyle{spphys}
\bibliography{refs_epjc}

\end{document}